\documentclass[a4paper,10pt]{article}
\usepackage{epsf}  
\usepackage{amsfonts, amssymb, amsthm, amsmath}  
\usepackage{tikz}
\usetikzlibrary{matrix,arrows,decorations.pathmorphing}
\usepackage{tikz-cd}
\usepackage{pdfsync}
\usepackage{graphicx}
\usepackage{hyperref}
\usepackage{psfrag}
\newtheorem{thm}{Theorem}  

\newtheorem{cor}[thm]{Corollary}  
\newtheorem{lemma}[thm]{Lemma}  

\newtheorem{remark}[thm]{Remark}  
\newtheorem{defn}[thm]{Definition}  
 
\newtheorem{heuristic}[thm]{Heuristic}

\newtheorem{example}[thm]{Example}  
\numberwithin{thm}{section}  
\def\pf{\noindent\emph{Proof: }}  
  
\def\stop{\hfill$\square$}  
\providecommand{\lrb}[1]{\ensuremath{\left(#1\right)}}

\providecommand{\abs}[1]{\left\lvert #1\right\rvert}  
\providecommand{\norm}[1]{\left\lVert #1\right\rVert}

\providecommand{\bra}[1]{<#1\vert}
\providecommand{\ket}[1]{ \vert#1>}

\DeclareMathOperator{\prob}{Prob}

\title{Natural probability}  

\usepackage{authblk}
\author{Brett Parker}
\affil{Mathematical Sciences Institute, The Australian National University}

\begin{document}

\maketitle

\begin{abstract}
How should we model an observer within quantum mechanics or quantum field theory? How can classical physics emerge from a quantum model, and why should classical probability be useful? How can we model a selective measurement entirely within a closed quantum system? This paper sketches a new physical theory of probability based on an attempt to model classical information  within a purely quantum system. 

We model classical information using a version of Zurek's theory of Quantum Darwinism,  with  emphasis on quantum information  encoded using projection operators localised in spacetime. This version of Quantum Darwinism is compatible with quantum field theory, and does not require any artificial division of a quantum system into subsystems. 
 
The main innovation is our attempt to provide a physical explanation of probability. Decoherence is the physical mechanism behind Quantum Darwinism or the `branching of quantum worlds'. Assuming a type of perfect decoherence we construct a conventional probabilistic model for classical information. This, however, is not our theory of natural probability, and does not quite demonstrate the validity of Bayesian reasoning.  Instead, our theory of natural probability arises from careful consideration of errors in decoherence: roughly speaking, we don't observe low probability events because they are swamped by quantum noise.

\end{abstract}
\tableofcontents
\section{Introduction}
How can we square our everyday experience with a quantum description of Nature? In particular, how should we model ourselves and our perceptions within quantum mechanics? 
The semi-classical limit links quantum mechanics with classical mechanics, but many classical aspects of how we perceive reality remain at odds with quantum mechanics. We intuitively assume that there is a single, completely describable, objective state of Nature, and that we could, in theory, discover anything about that state. This assumption is not compatible with quantum mechanics, and can not be derived using a semi-classical limit.
 The semiclassical limit also does not explain another important aspect of our observed classical reality; namely probability theory.  Bayesian reasoning is useful, and we ultimately rest empirical science on the assumption that usual probability theory is valid.  The goal of this paper is to provide a physical explanation for probability. 

The mathematical formulation of probability features an outcome space $X$, a $\sigma$-algebra of events, $E\subset X$, and a probability measure $\mu$ assigning a probability $\mu(E)\in[0,1]$ to each event. This serves as a classical model for information, with logical operations reflected in the algebra structure of events; for example,  the logical combination `$E_1$ and $E_2$' is represented by $E_1\cap E_2$. Starting with a quantum model of Nature, we approximate such a classical model and show that classical logic, probability and Bayesian reasoning emerge from our quantum model. We do more than this.  Our emergent model for classical information has more structure than usual probability theory. As well as a probability, each event has a quantitative measure of how visible it is in each region of spacetime, and we also have an extra notion of an event being `classical' and visible in many spatially separated regions of spacetime. We thus arrive at a new theory of natural probability, where a version of the statement `you will never perceive a low probability event' is actually valid, instead of being a setup for a common fallacy starring an archer who paints a target after shooting his arrow.

We assume a quantum model of Nature, with a notion of observables being localised in a spacetime region, and we also assume standard causality axioms; quantum field theories --- our best verified models of Nature --- have this feature. We do not assume any version of the controversial\footnote{Apart from violating causality, a problem with wavefunction collapse is that it is a separate process, without a well-defined specification of exactly when or where it happens. This makes a quantum theory with wavefunction collapse an inexact theory, even though, when used with judgement and good taste,  it is Fine For All Practical Purposes; \cite{Bell_1990}.  A similar objection applies to DeWitt's overstated `multiple world' popularisation of Everett's relative state intepretation of quantum mechanics; \cite{DeWitt}.  See  \cite{evidence} for a survey of the experimental non-evidence of wavefunction collapse as a physical process.} wavefunction--collapse postulate, so must work hard to reproduce the undeniable experimental consequences of selective measurements.  Within such a quantum model of Nature, we model classical events  using certain projection operators, called classical projection-statements; see Definition \ref{classical def}. So, we model all classical information and all classical objects --- including classical observers, experimental setups, and any unfortunate cats ---  as classical projection-statements. These classical projection-statements have the property that they are visible in many separate regions of spacetime, just as classical information can be discovered in many separate regions of spacetime. The thesis we propose is as follows: \emph{our perceptions of reality are described by classical projection-statements}.    From this extra assumption, we explain how Bayesian reasoning emerges as approximately valid, sketching our new theory of natural probability.

 This theory should be regarded as an extension of Everett's relative state formulation of quantum mechanics, \cite{Everett}, and Zurek's Quantum Darwinism, \cite{QD}, with the main new content our treatment of probability. The role of probability in no-collapse quantum mechanics has been the subject of much research. There are many beautiful approaches to deriving the Born--Vaidman rule using symmetry \cite{Deutsch, WallaceBorn, Zurek_2005,Sebens_2018}, and like Gleason's original derivation of the Born rule, \cite{Gle}, these generally find that the quantum situation is nicer than classical probability, because reasonable constraints lead to a unique probability measure. Our focus and approach is different. We are instead asking \emph{how can classical probability emerge from a quantum model, and is the usual mathematical formuation of probability the best formulation of what emerges?} Our approach is  based on observing the essentially \emph{approximate} nature of the decoherence behind Quantum Darwinism and the `branching of worlds' in Everettian quantum mechanics.  As such, natural probability could be regarded as a version of Hanson's suggestion, \cite{Hanson_2003}, that observed frequencies  in a multi-world quantum model can be explained by positing that there are no observers in worlds with small measure, because these small worlds are `mangled' by inexact decoherence. 

\

In Section \ref{projection-statements}, we define a projection-statement to be an orthogonal projection localised in a specified region of spacetime. The visibility of a projection-statement depends on the quantum state $\psi$ of Nature. In particular, a projection-statement $P$ is visible in a spacetime region $U$, if there is a projection-statement $P'$ such that $P'\psi$ approximates $P\psi$. So, visibility measures some aspect of entanglement. We  define a classical projection-statement as a projection-statement, localised in U,  that is also visible in many spatially separated regions in the causal future of $U$.  This unfortunately imprecise definition is tolerable for describing an emergent property, because we are not suggesting any modification of physics that depends on this definition.

\
%
%
%
%
%
 We regard a projection-statement  as `something one could say about Nature', without insisting it be `really' true or false, and regard a classical projection-statement as `something one could say about Nature that has many consequences, and that many separate observers could discover'.  
%
%
%
%
%
%
%
In Section \ref{logical combination} we explain how to associate operators and `relative probabilities' to logical combinations of projection-statements, so long as these projection-statements  can be time-ordered. In this formalism, if the projection-statement $O$ is in the future of $P$, we can say that $O$ `knows' $P$ if $OP\psi\approx O\psi$. Similarly,  if $O$ is in the past of $P$, we can say that $O$ can predict $P$ if $PO\psi\approx O\psi$.  For arbitrary projection-statements, this is a bad model of logic, however, we will show, in sections \ref{precise} and \ref{approximate}, that we get an approximate model for logic and probability when we restrict to using classical projection-statements.

Sections \ref{precise} and \ref{approximate} introduce our model for classical probablistic reasoning with classical projection-statements. If some collection of classical projection-statements have a collection of spatially separated records $P'_i$, there is a Boolean algebra $\mathfrak B$ of projection-statements constructed from logical operations in $P_i'$. This Boolean algebra also has the structure of a $\mathcal H$--valued measure algebra, where $\mathcal H$ indicates the Hilbert space\footnote{In the case of a quantum field theory, we use the GNS representation \cite[Section 3.2]{AQFT} to construct an appropriate Hilbert space $\mathcal H$  and state $\psi\in\mathcal H$ from a state of the quantum field theory.} of our quantum theory. This $\mathcal H$--valued measure algebra is defined by a map $\phi:\mathfrak B\longrightarrow \mathcal H$ sending $P_i'$ to $P'_i\psi\in\mathcal H$. Such $\mathcal H$--valued measure algebras (defined in Section \ref{ma section}) are an intermediate notion between projection-valued measures and classical probability spaces. The important point is that the $\mathcal H$--valued measure algebra $(\mathfrak B,\phi)$ is approximately independent of the choice of records $P_i'$; as shown in Corollary  \ref{totally classical BA} and Heuristics \ref{separated record replacement} and \ref{algebra record replacement}.  Moreover, Lemma \ref{space from algebra} provides a canonical construction of a probability space from a $\mathcal H$--valued measure algebra. So, we obtain an approximately canonical probabilistic model for logical operations in projection-statements.

In Section \ref{precise}, we show how classical information can emerge if we assume the existence of totally classical projection-statements.   In particular, we construct a canonical probability space $(X,\mu)$ together with a collection of  $\sigma$--algebras $\mathfrak F_O$ representing classical information obtainable in open subsets $O$ of spacetime.  We don't actually expect totally classical projection-statements to exist, so in Section \ref{approximate}, we show how to obtain approximations to this situation, in the form of classical $\mathcal H$--valued measure algebras.

In Section \ref{Decoherence}, we sketch how decoherence can produce classical projection-statements, and arrive at some heuristics about what we can realistically expect from classical projection-statements. It is not true that every conceivable quantum mechanical system will have classical projection-statements, so we only argue that some systems have classical projection-statements, and give some simplified examples of how they can occur. An important point of Section \ref{Decoherence} is that decoherence in realistic systems is not perfect. On the one hand, this is inconvenient as it makes precise mathematical definitions and arguments more difficult, but  this deficit in decoherence forms the basis for our theory of natural probability.

 In Section \ref{observer}, we 
explain how to model  observers and how observers can detect classical projection-statements. Crucially, we can only expect a projection-statement $P$ to be \emph{approximately} visible to an observer. An observer may have direct access only to $P'$ where $\norm{(P-P')\psi}$ is small. For the projection-statement $P'$ to meaningfully encode information about $P$, we require that this error $\norm{(P-P')\psi}$ is small compared to $\norm{P\psi}$. From heuristic arguments in Section \ref{Decoherence}, we expect a lower bound $\norm{(P-P')\psi}\geq p_{0}$ where $p_{0}$ depends on physical properties of the observer and $P$ independent of $\norm{P\psi}$. So, for $P$ to be visible to our observer we require $ \norm{P\psi}\gg p_{0}$. Accordingly, for $P$ to be a classical projection-statement  visible in many separate regions of spacetime, we need $\norm{P\psi}$ not too small. 

In Section \ref{NP section}, we outline our theory of natural probability, and explain how low-probability classical projection-statements are harder to detect, and require more information to reliably describe.  In Section \ref{measurement}, we describe how the process of quantum measurement looks from our perspective, carefully examining the question of when a density matrix can be regarded as a statistical ensemble of states. We also  compare our approach to several popular interpretations of quantum mechanics, finding that most have partial validity. We pay special attention to the consistent or decohered histories approach advocated by Griffith, Omn\`es, and Gell-Mann and Hartle, and to Zurek's theory of Quantum Darwinism.

\

\

\section{ Projection-statements}

\label{projection-statements}

\

In this paper, we assume\footnote{A complete description of Nature should also involve gravity, so this assumption is false. Nevertheless, we could hope that it is approximately true, and that there is a notion of localised projection-statements and information about their relative positions within a complete quantum theory of Nature. See Section \ref{gravity} for further discussion.} a fixed spacetime, such as Minkowski space. We also assume a quantum model of Nature, featuring a Hilbert space $\mathcal H$ and a unit vector $\psi\in\mathcal H$ describing the state of nature. As standard in quantum field theory, we assume a connection between spacetime and $\mathcal H$ is provided by a notion of local observables, so we have a notion of an observable being localised in a particular spacetime region. 
 
Our projection-statements are meant to play the role of `local, and physically detectible, things one could say about Nature'.  Quantum field theory features idealized  local observables  $O(x,t)$  that commute when spatially-separated. To obtain an observable $O$, we can smear out $O(x,t)$ over some spatial region $U$.  An example of a projection-statement localized in $U$ is the projection operator corresponding to a range of values of $O$. Such localized projection-statements commute when they are spatially-separated.  For simplicity, we do not work directly with quantum field theory in this paper, but instead we assume that our quantum system has some comparable notion of a projection-statement $P$ localised in a region $U_P$. One tractable  example of such a quantum system is a lattice approximation of quantum field theory.  In general, any algebraic quantum field theory can also be formulated in a Hilbert space using the GNS representation \cite[section 3.2]{AQFT}.

\begin{defn}A projection-statement $P$ is a projection operator $P$ on $\mathcal H$ together with a spacetime region $U_P$ where $P$ is localised.
\end{defn}

We will use standard definitions of the causal relations between subsets of spacetime, and apply these to projection-statements.
\begin{defn} Say that a projection-statement $P$ is spatially separated from $Q$ if every point in the region $U_P$ where $P$ is localised is spatially separated from every point in the region $U_Q$ where $Q$ is localised. Similarly, say that $P$ is in the (causal) future of $Q$ if every point in $U_P$ is in the (causal) future of all points of $U_Q$. \end{defn}

We assume the following causality axioms for projection-statements, analogous to standard axioms in algebraic quantum field theory.
\begin{itemize}
\item[\bf{Isotony}] If $U'\subset U$ and $P$ is localised in $U'$, then it can  also be localised in $U$.
\item[\bf{Einstein causality}] Whenever $P$ and $Q$ are spatially separated, the projection operators $P$ and $Q$ commute.
\item[\bf{Time slice}] If $U'$ contains a Cauchy surface\footnote{A Cauchy surface $S$ for $U_P$ is a space-like surface with the property that every inextensible causal curve in $U_P$ intersects $S$ exactly once.} for $U_P$,   then $P$ can also be localised in $U'$.
\item[\bf{Algebraic operations}] If $P$ is localised in $U$, then the complementary projection $1-P$ is also localised in $U$. Moreover, if $P$ and $Q$ commute and are both localised in $U$, then the projection $PQ$ is also localised in $U$.

\end{itemize}

\begin{example}\label{QFT} Algebraic field theory, \cite{AQFT}, provides physically realistic examples of quantum systems with projection-statements satisfying the above axioms. An algebraic field theory features a unital $*$--algebra $\mathcal A$, analogous to the algebra of bounded operators on $\mathcal H$. As part of the structure, we also have  sub-algebras $\mathcal A(U)\subset \mathcal A$ associated to open subsets $U$ of spacetime, satisfying appropriate  causality axioms. Within algebraic field theory, a state is a linear functional \[\sigma:\mathcal A\longrightarrow \mathbb C\] that is semi-positive in the sense that $\sigma(A^*A)\geq 0$, and normalised so that $\sigma(1)=1$. So,  $\sigma(A)$ is analogous to $\bra\psi A\ket \psi$, and $\sqrt{\sigma(A^*A)}$ plays the role of $\norm{A\psi}$. The notion of an  orthogonal projection in an algebraic field theory is an element  $P\in\mathcal A$ that satisfies $P=P^*$ and $P^2=P$, and we can say that this projection is localised in $U$ if $P\in \mathcal A(U)$.

Given a state $\sigma$, the Gelfand-Naimark-Segal construction  gives a state $\psi$ in a Hilbert space $\mathcal H$, and a representation of $\mathcal A$ as an algebra of linear operators on $\mathcal H$ such that $\sigma(A)=\bra\psi A\ket \psi$;  \cite[section 3.2]{AQFT}. So, we don't lose generality by working in a Hilbert space model. 
\end{example}

\begin{defn} \label{visible} Say that a projection-statement $P$ is visible in a region $U$ if there is a projection-statement $P'$ localized in $U$ so that 
\[\frac{\norm{ (P-P')\psi}}{\norm{P\psi}}\ll 1\ \ .\]
In the case that $U$ is in the causal future of $U_P$, say that $P'$ is a record of $P$.

Define the visibility of $P$ in $U$ to be
\[\text{Vis}_U(P):=\begin{cases}\sup_{P'}\ln\norm{P\psi}-\ln\norm{(P-P')\psi} &\text{ if }P\psi\neq 0
\\ 0 &\text{ if }P\psi=0\end{cases}\]
where the supremum is taken over projection-statements $P'$ supported in $U$.

\end{defn}

So, $P$ is visible in $U$ if there is a projection-statement $P'$ localised in $U$ so that $P\psi\approx P'\psi$. The visibility of $P$ in $U$ takes values in $[0,\infty]$  with higher numbers indicating more visibility. If we image observers using particular sensors, we could also restrict Definition \ref{visible} to only use projection statements $P'$ of a particular type, determined by the sensors used. As it is, we stick to Definition \ref{visible} because this weaker version is all that is needed for the approximate validity of usual probability theory. 

If a projection-statement is visible in a space-time region $U$, we can think that it is discoverable by observation within $U$. We expect observers themselves to be local, so it reasonable to focus attention on such locally-discoverable information.  As emphasized by Zurek in \cite{QD, Zurekexistential}, one property of classical events is the following: they can be independently discovered by many separate observers. We shall loosely
   define a classical projection-statement $P$ as follows:

\begin{defn}\label{classical def} A classical projection-statement $P$ is a projection-statement that
     is visible in many spatially separated regions within the causal future of $P$.  \end{defn}
So, a classical projection-statement is one with many spatially separated records. The records themselves will also tend to be classical projection-statements. 

\begin{remark}The reader will note that a projection-statement being visible in $U$ is an approximate property, which can hold to a greater or lesser degree.  Our definition of being classical is also approximate, but significantly vaguer, because we need to choose `many' different regions to measure the visibility of a projection-statement.   It is not necessary for an emergent property to have a precise, easy-to-state   definition. We are not proposing any separate fundamental physical process that depends on projection-statements being classical, just as we are not suggesting that there is a sharply defined set of classical projection-statements that are `really true'. This is in contrast to an intuitive conception of a unique objective classical reality. 

In section \ref{precise}, we see the convenient consequences of more precise (but unrealistic) versions of Definitions \ref{visible} and \ref{classical def}.
\end{remark}

\subsection{Logical combinations of projection-statements}
\label{logical combination}

If we consider  projection-statements as `things we can say about Nature', we need a way of interpreting logical combinations of projection-statements, ideally in the form of projection-statements.

\begin{defn} For a projection-statement $P$, define `not $P$' to be the projection-statement 
\[\neg P:=1-P\]
where $1$ indicates the identity operator on $\mathcal H$. This projection-statements is localised in the same region that $P$ is localised, so $U_P=U_{\neg P}$.
\end{defn}

\begin{defn} If $P$ and $Q$ are spatially separated or $P$ is in the future of $Q$, define `$P$ and $Q$' as the time-ordered composition of $P$ and $Q$. 
\[P\wedge Q=Q\wedge P:=PQ\]
In the case that $P$ and $Q$ are spatially separated, define the projection-statement $P\wedge Q$ as the above operator localised in $U_P\cup U_Q$.
\end{defn}

If $P$ and $Q$ are not spatially separated, they need not commute, so $P Q$ need not be an orthogonal projection, and is hence not a projection-statement.\footnote{In the case that $P$ and $Q$ don't commute, the reader may be tempted to suggest defining $P\wedge Q$ as projection to the linear subspace given by the intersection of the ranges of $P$ and $Q$. This definition is not suitable for our purposes as it is not stable under perturbations of $P$ and $Q$.} However, if $Q$ is recorded in a projection-statement $Q'$ spatially separated from $P$,   Lemma \ref{replacement} implies that $P\wedge Q\psi\approx P\wedge Q'\psi$, so we can use the projection-statement $P\wedge Q'$ as an approximation of  the logical statement `$P$ and $Q$'. Similarly, if $P$ and $Q$ are classical projection-statements, with spatially separated records $P'$ and $Q'$, the projection-statement $P'\wedge Q'$ is an approximate notion of `$P$ and $Q$', even if $P$ and $Q$ are localised in the same region. For classical projection-statements $P$ and $Q$, we will see that this replacement is largely independent of the choice of records, $P'$ and $Q'$; see Lemmas \ref{exact replacement} and \ref{separate representatives}, Corollary \ref{totally classical BA}, and Heuristics \ref{separated record replacement} and \ref{algebra record replacement}.

Every logical operation can be written as a combination of $\neg$ and $\wedge$. For example `$P$ or $Q$' is 
\[ P\vee Q=\neg(\neg P \wedge \neg Q)\ \]
and, when $P$ is in the future of $Q$, or spatially separated from $Q$, we identify this with the operator
\[Q\vee P=P\vee Q=1-(1-P)(1-Q)=P+Q-PQ\]

\begin{defn} Suppose that we have a logical combination of projection-statements $P_i$, and these projection-statements are mutually time separated or spatially separated. Associate an operator $E$  by writing this logical combination using $\neg $ and $\wedge$, then time-ordering the resulting expression in the $P_i$ so that $P_i$ occurs to the left of $P_j$ whenever $P_i$ is in the future of $P_j$.
\end{defn}

So, we can associate an operator $E$ to a logical combination of projection-statements $P_i$ if they are mutually time separated or spatially separated. If the original projection-statements are spatially separated, then $E$ is an orthogonal projection, and hence is also a projection-statement, localised in the region $\bigcup_i U_{P_i}$. If we have a logical combination of classical projection-statements with spatially separated records $P_i'$, we can also associate the projection-statement constructed from the logical combination of the records $P_i'$. We shall see that this projection-statement is approximately independent of the choice of these records, with the caveat that the approximation is worse for more complicated logical combinations. 

We will assign `probabilities' to logical combinations of projection-statements, with scare quotes indicating that, for now, these are just numbers without any particular interpretation.
\begin{defn} Define
\[\prob(P):=\norm{P\psi}^2\]
and
\[\prob (P\vert Q):=\frac{\norm{P\wedge Q \psi}^2}{\norm{Q\psi}^2}\ .\]
Similarly, if $E$ is a logical combination of projection-statements, define
\[\prob(E):=\norm{E\psi}^2\]
\end{defn}

The `probabilities' assigned to logical combinations of time-separated projection-statements have no reason to obey the usual laws of probabilities. Nevertheless,  we will show that this is often approximately true if $P_i$ are classical projection-statements;  see Lemma \ref{exact replacement} and Heuristics \ref{separated record replacement} and \ref{algebra record replacement}. We will also show that probabilities obtained using time-separated records of classical projection-statements are approximately independent of the records chosen.

Even if a quantum system has a notion of localized projection-statements, there is no reason to believe that classical projection-statements will always exist.   Section \ref{Decoherence} gives a heuristic argument that decoherence processes  can produce classical projection-statements. If such decoherence happens because of interactions with particles traveling at the speed of light, we can expect such a classical projection-statement $P$ to be visible in many different regions within the  future of $P$, including some connected to $U_P$ by light rays. For such  classical projection-statements, the geometry of 4-dimensional Minkowski spacetime ensures that we can expect the hypotheses of the following lemma to often apply.  With this lemma, we shall  see that Bayesian reasoning about such classical  projection-statements is approximately valid, because logical expressions in these projection-statements can be approximated by logical expressions in commuting projection-statements $P_i'$.

\begin{lemma}\label{replacement} Suppose that the projection-statements  $P_{1},\dotsc, P_{n}$ and $P_{1}',\dotsc,P_{n}'$  have the following properties:
\begin{enumerate}
\item \[\norm{ P_{i}\psi- P_{i}'\psi}\leq\epsilon_{i}\]
\item  for all $j>i$, the projection-statement $ P_{j}$ is in the past of $P_{i}$, or $P_i$ is spatially separated from $ P_{j}$;
\item for all $j>i$  the projection--statement  $ P'_{j}$ is spatially separated from both $ P_{i}$ and $ P'_{i}$.
\end{enumerate}
For any logical expression $E$ in the projection-statements $ P_{i}$, let $E'$ be the corresponding logical expression in the replacement projection-statements $ P_{i}'$. Then
\[\norm{(E-E')\psi}\leq 2^{n}\sum_{i}\epsilon_{i}\ .\]
Moreover, a sharper bound can be obtained as follows: Write $E=\sum_j E_j$, where each $E_j$ is a time-ordered product of terms in the form $P_i$ or $(1-P_i)$. Then, let $m_i$ be the number of times $P_i$ appears in this expansion of $E$. Our sharper bound is 
\[\norm{(E-E')\psi}\leq \sum_{i} m_{i}\epsilon_{i} \ .\]
\end{lemma}

 \begin{figure}[htb]
\includegraphics[width=\textwidth]{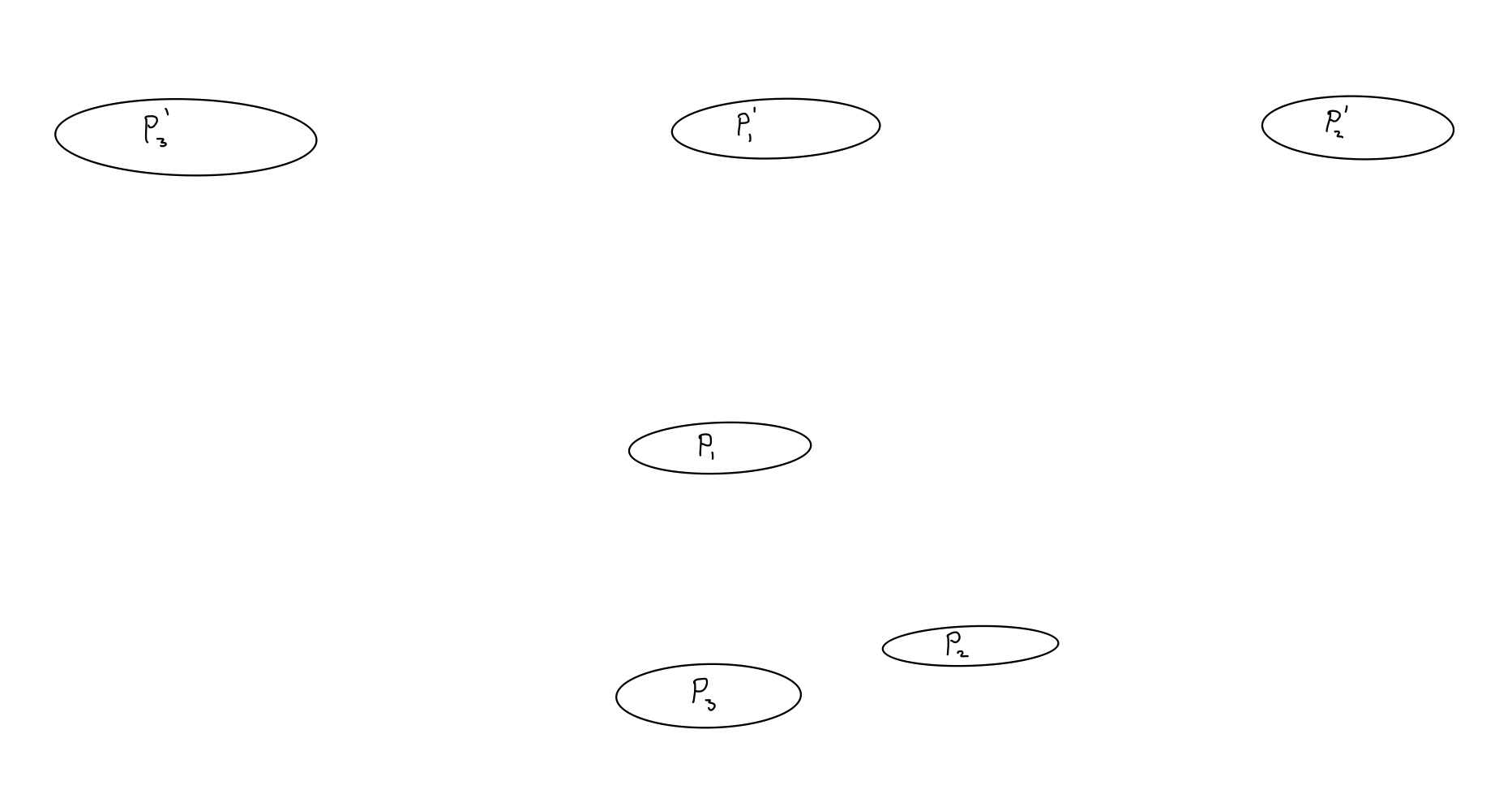} 
\caption{Spacetime location of projection-statements satisfying the requirements of Lemma \ref{replacement}. The vertical axis is time, and we use coordinates in which the speed of light is 1. Note that there is more space to satisfy the requirements of Lemma \ref{replacement} in 4-dimensional Minkowski space, when compared to this 2-dimensional picture. } 
\label{fig1}
\end{figure}

\pf
Suppose first that $E=\prod_{i\in I} P_{i}$.
Then we can expand $\norm{(E-E')\psi}$ as follows.
\begin{equation}\label{E estimate}\begin{split}\norm{(E-E')\psi}&=\norm{\sum_{j\in I}\lrb{\prod_{i<j,\  i\in I} P_{i}}( P_{j}- P_{j}')\lrb{\prod_{k>j,\ k\in I} P_{k}'} \psi} 
\\ &= \norm{\sum_{j\in I}\lrb{\prod_{i<j,\  i\in I} P_{i}}\lrb{\prod_{k>j,\ k\in I} P_{k}'}( P_{j}- P_{j}') \psi} 
\\ & \leq \sum_{j\in I}\epsilon_{j}\end{split}\end{equation}
The second line is allowed because $P_{k}'$ is spatially separated from both $P_j$ and $P_j'$. The same inequality holds if we replace some of the $ P_{i}$ by $(1- P_{i})$. More complicated logical expressions $E$ may be written (non-uniquely) as a sum of terms in this form, so for these $E$, 
\[\norm{(E-E')\psi}\leq \sum m_{i}\epsilon_{i} \]
where $ P_{i}$ or $(1- P_{i})$ appears in $m_{i}$ terms in our sum expressing $E$. Any logical expression $E$ in the $ P_{i}$ can be written uniquely as a sum of terms, each involving either $ P_{i}$ or $(1- P_{i})$ for all $i$ between $1$ and $n$, and there are $2^{n}$ possible such terms, so a worst-case bound for $\norm{(E-E')\psi}$ is $2^{n}\sum_{i}\epsilon_{i}$.

\stop

\begin{remark}In deriving inequality (\ref{E estimate}), we used that $ P'_{k}$ and $ P_{i}$ are  projections, and hence have norm bounded by $1$. In many situations, applying these projections will significantly reduce the size of the error term,  $( P_{i}- P'_{i})\psi$.  Accordingly, the actual expected error could be significantly smaller than the upper bound from (\ref{E estimate}).
\end{remark}
\begin{remark} The assumption that $P_k'$ is spatially separated from $P_j$ for $k>j$ is a strong assumption, motivated by imagining that $P_k$ is visible in some regions connected to $U_{P_k}$ by light rays; see Figure \ref{fig1}. The conclusions of Lemma \ref{replacement} still hold if we replace this assumption by the assumption that \[\norm{(P_j-P_j')\left(\prod_{k>j}P_k'\right)\psi}\leq \epsilon_j \norm{\left(\prod_{k>j} P_k'\right)\psi} \ .\] This alternate assumption could be justified by imagining that $P_j'$ encodes an observation that provides robust evidence for $P_j$.
\end{remark}

The importance of Lemma \ref{replacement} is the following. If we had that $\norm{ (E-E')\psi}=0$ for all logical statements $E$ in $\{P_{i}\}$, and corresponding logical statements $E'$ in $\{P'_{i}\}$, then we would also have that $\prob (E)=\prob (E')$. The projection-statements $P_{i}'$ are spatially-separated, and therefore commute, so the numbers $\prob(E')$ obey the usual rules of probability theory, and the same must therefore hold for $\prob (E)$. So,  there would be a space $X$ with a probability measure $\mu$ and  with subsets $X_{E}$ so that $X_{E_{1}\wedge E_{2}}=X_{E_{1}}\cap X_{E_{2}}$, $X_{\neg E}=X\setminus X_{E}$, and $\prob(E)=\mu(X_{E})$. So,  Lemma \ref{replacement} is telling us that the probabilities we assign  approximately obey normal probability theory. 

 The deviation from classical probabilistic reasoning is more serious when dealing with projection-statements  that are less probable: For example suppose that $\prob(B\vert A)=1$ and $\prob(C\vert B)=1$, then classically, one would expect that $\prob(C\vert A)=1$. However,  if all we know is that logical statements in $\{A,B,C\}$ approximately obey normal probability with error $\epsilon$, all one can infer is that $\prob(C\vert A)$ is somewhere between $1$ and approximately $1-2\epsilon/\prob(A)$. So, predictions using $A$ become worse when $\prob A\ll 1$.
This issue persists even if we assume a closer approximation of classical probability.  
 Even in the case that $\abs{\prob E-\prob E'}<\epsilon \prob E$,  the possible deviation from the conclusion  $\prob(C\vert A)=1$ is approximately  $(1+\prob(B)/\prob(A))\epsilon$, and this deviation is worse  when $\prob(A)$ is small. To reiterate:  even if $A$ implies $B$ and $B$ implies $C$, we can not reliably conclude that $A$ implies $C$ when $\prob(A) \ll \prob (B)$.

 \
 
If $P$ is a record of a classical projection-statement $P_0$, we can expect that $P$ is also visible in many regions spatially separated from $P$. Suppose we pick a Cauchy surface to have a notion of simultaneous projection statements, and choose some collection of classical projection-statements with records on this Cauchy surface. The next lemma can be used to show that logical combinations of these records are approximately independent of the records chosen.   This will be important for concluding that we can assign probabilities using records of classical projection-statements, with these probabilities approximately independent of the records chosen.

 \begin{lemma} \label{record replacement}Suppose that $\norm{(P_i-P_i')\psi}<\epsilon_i$, and that $P_i'$ is spatially separated from $P_j$ and $P_j'$ for all $j\neq i$. Then
 \[\norm{(P_1\dotsc P_n-P_1'\dotsc P_n')\psi}\leq\epsilon_1+\dotsc +\epsilon_n\]
 and, for any permutation $\sigma$ of $\{1,\dotsc, n\}$
 \[\norm{(\prod_i P_i-\prod_i P_{\sigma(i)})\psi}\leq 2\sum_i\epsilon_i\]
\end{lemma}  
\begin{proof} In  the proof of Lemma \ref{replacement}, the condition that $P_j$ is in the past of $P_i$ when $i<j$ is used only to order logical expressions in the $P_i$. The other conditions of Lemma \ref{replacement} hold, so applying the same argument, we conclude
\[\norm{(\prod_iP_i-\prod_i P_i')\psi}\leq \sum_i\epsilon_i\]
and the same holds when we permute the order of our projections.

We have that the projections $P_i'$commute because they are spatially separated, so we get our required result: for any permutation $\sigma$
\[\norm{(\prod_iP_i -\prod_i P_{\sigma(i)})\psi}\leq 2\sum_i \epsilon_i\]

\end{proof}

\begin{remark}In Zurek's Quantum Darwinisim, there is a division of $\mathcal H$ into a tensor product $\bigotimes_{i\in I}\mathcal H_i$ of separate subsystems, and Zurek demands that classical information is accessible by sampling a small proportion $J\subset I$ of these subsystems,  and is accessible in many such fragments of $\mathcal H$; \cite{QD}.  The analogue of a projection-statement in this setting is a projection operator $P$  from a subsystem $J_P\subset I$, so induced from a projection on $\bigotimes_{i\in J_P}\mathcal H_i$.  The analogue of projection-statements $P$ and $Q$ being  spatially separated is then $J_P\cap J_Q=\emptyset$, so $P$ and $Q$ act on separate subsystems, and commute just like spatially separated projection-statements.   A version of Lemma \ref{record replacement} also holds in this setting. As a consequence, many of our later arguments can also be applied in the setting of Quantum Darwinism.
\end{remark}
\section{The emergent structure of classical information}
\label{precise}

What is an idealised model for classical information on a fixed spacetime? Such a model features a set $X$ of possible complete histories of Nature, along with a probability measure $\mu$ on $X$. Implicit in $\mu$ is a $\sigma$-algebra $\mathfrak F$ of measurable subsets of $X$, that is, a collection $\mathfrak F$ of subsets of $X$ closed under the operations of taking complements and countable intersections or unions. Within probability theory, these measurable subsets are sometimes called events. To represent information available in an open region $U$ of spacetime,  we have a sub $\sigma$-algebra $\mathfrak F_U\subset \mathfrak F$, so that the observable quantities within $U$ consist of functions on $X$ that are measurable with respect to $\mathfrak F_U$. So, within $U$, we can theoretically determine whether the state of Nature is within a set $E$ in the $\sigma$-algebra $\mathfrak F_U$.

 So, a model for classical information consists of $(X,\{\mathfrak F_U\},\mu)$, where $X$ is a set of classical possibilities, $\mathfrak F_U$ is a $\sigma$--algebra representing classical information obtainable in the space-time region $U$, and $\mu$ is a probability measure. 

Reasonable properties we could classical expect include the following. 
\begin{itemize}
\item[I]  If $U\subset V$, then $\mathfrak F_U\subset \mathfrak F_V$.
\item[II] If the past Cauchy development of  $V$ contains  $U$, then  $\mathfrak F_U\subset \mathfrak F_V$. 
\item [III] \label{III} If the causal past of $V$ contains $U$, then $\mathfrak F_U\subset \mathfrak F_V$. 
\end{itemize}
The first property encodes our expectation that if $U\subset V$, the information obtainable in $V$ contains the information obtainable in $U$. The second property encodes the expectation that locally obtainable information in $U$ should be obtainable in the future, \emph{somewhere} causally connected to $U$. The third property encodes the optimistically strong expectation that locally obtainable information within $U$ should be obtainable in the future, \emph{everywhere} causally connected to $U$. Note that the third property implies the second, which implies the first.

 In this section, we construct $(X,\{\mathfrak F_U\}, \mu)$ satisfying the above conditions. We use an extreme version of classical projection-statements, called totally classical projection-statements. We first show that equivalence classes of totally classical projection-statements compactly supported within $U$ form a Boolean algebra $\mathfrak B_U$ with a natural finitely additive probability measure $\mu:\mathfrak B_U\longrightarrow [0,1] $, and also a natural map $\phi:\mathfrak B\longrightarrow\mathcal H$ with $\mu(b)=\norm{\phi(b)}^2$. We call  $(\mathfrak B,\phi)$ a $\mathcal H$--valued measure algebra. We show that each $\mathcal H$--valued measure algebra has a natural completion $(\bar{\mathcal B},\phi)$ corresponding to a real  measure algebra $(\bar{\mathfrak B}_U,\mu)$. We then construct $(X,\{\mathfrak F_U\},\mu)$ so that $(\bar{\mathfrak B}_U,\mu)$ is the quotient of $(\mathfrak F_U,\mu)$ by the ideal of measure $0$ subsets. 
 
 In this way, we obtain an idealised model of classical information. The problem is, we don't expect totally classical projection statements to exist in realistic quantum models, so Section \ref{approximate} is devoted to deriving an approximation of this idealised situation.

 To obtain the results in this section, we assume  that our spacetime is complete, and we assume that the causal future of closed sets is closed, the causal future of open sets is open,  and that the chronological future of a set is the interior of the causal future.  These assumptions hold in Minkowski space.

\begin{defn} A projection-statement $P$ is totally visible in an open region $U$ if there is a projection-statement $P'$ localised in $U$ such that $P\psi=P'\psi$.
\end{defn}
 \begin{defn} A projection-statement $P$ is compactly supported within an open set $U$ if it is localised in a region $U_P$ compactly contained in $U$.
 \end{defn}

 \begin{defn} A projection-statement $P$ is totally classical if it is localised in a bounded open region $U_P$, and is totally visible in every open region $U$ such that  $P$ is compactly supported within the causal past of $U$.

 \end{defn}
 
One could justify the definition of a totally classical projection-statement by imagining that classical information about an event can be observed anywhere in the future causally connected to the event. This is unrealistic, but also close to a classical conception of information.

 \begin{defn}Two totally classical projection-statements $P$ and $P'$ are equivalent if $P\psi=P'\psi$\end{defn}

Note that each totally classical projection-statement $P$ is equivalent to a totally classical projection statement localised in $U$, where $U$ is any bounded open region whose causal past compactly contains $U_P$.

 For the following lemma, we say that a projection-statement is localised on an open subset $U$ of a space-like hypersurface if it is localised within the Cauchy development of $U$. 
 
 \begin{lemma} \label{exact replacement}Suppose that $P_1,\dotsc,P_n$ are totally classical projection statements, and that, for all $i<j$, we have $U_{P_i}$ compactly supported in the future of each point in $U_{P_j}$. Suppose further that there is a spacelike hypersurface $S$ in the future of all $U_{P_i}$ so that the causal future of $U_{P_i}$ intersects $S$ in a nonempty open subset that does not contain any connected component of $S$.   Then,  there exist spatially separated projection statements $P_i'$ localised on $S$  such that $P_i$ is equivalent to $P_i'$, and such that
 \[P_1\dotsb P_n\psi=P_1'\dotsb P_n'\psi\]
 Moreover, given any logical expression  $E$ in the projection-statements $P_i$,  we have that
 \[E\psi=E'\psi\]
 where $E'$ is given by the corresponding logical expression in the projection-statements $P_i'$. 
 \end{lemma}
 
 \begin{proof}  Use $I^\pm(A)$ for the chronological future or past of a set $A$, and $J^{\pm}(A)$ for the causal future or past. For each point $x\in \bar U_{P_j}$ , we have that $I^+(x)$ contains $\bar U_{P_i}$ for all $i<j$, so  $I^+(x)$ contains  $J^+(\bar U_{P_i})$. 
 (Our assumptions on the causal structure of spacetime imply that $J^+(I^+(x))$ is an open subset of $J^+(x)$, and hence contained in $I^+(x)$.)
  In particular, $I^+(x)\cap S$ is nonempty and open, and  contains the closed set $J^+(\bar U_{P_i})$. By assumption,  $I^+(x)$ also does not contain any connected component of $S$, so it is not closed. We can therefore conclude that  
 \[S\cap I^+(x)\setminus\lrb{\bigcup_{i<j} J^+(\bar U_{P_i})}\] is a nonempty open subset of $S$.  This holds for all $x\in \bar U_{P_j}$. 
 Define the open subset 
 \[U_{P'_j}:=S\cap I^+(\bar U_{P_j})\setminus \lrb{\bigcup_{i<j} J^+(\bar U_{P_i})}\ .\]
 As argued above,  the chronological past of $U_{P'_j}$ contains each point $x\in \bar U_{P_j}$, so the totally classical projection-statement $P_j$ is equivalent to some totally classical projection-statement $P'_j$ localised within $\bar U_{P_j}$. Moreover, we have that  $P'_j$ is spatially separated from $P_i$ for all $i<j$, and also spatially separated from $P'$. The assumptions of Lemma \ref{replacement} therefore hold (with $\epsilon_i$ arbitrarily small), so we get that 
 \[P_1\dotsb P_n\psi=P_1'\dotsb P_n'\psi\]
 and similarly, given any logical expression $E$ in the $P_i$, we have that
 \[E\psi=E'\psi \ .\]

 \end{proof}
 
 Lemma \ref{exact replacement} tells us that, so long as our totally classical projection-statements can be time-ordered, the probabilities we assign to logical expressions in these totally classical projection-statements follow the laws of classical probability. This provides some justification for the probabilities we assigned in Section \ref{logical combination}. In what follows, we will also show that we can coherently assign probabilities  to arbitrary logical expressions in totally classical projection-statements. We will do this by showing that we can always choose spatially separated records, and showing that the probabilities are independent of the records chosen. For this, we will require the following two lemmas about the causal structure of spacetime. The proof of these two lemmas was pointed out to me by James Tener. 
 
 \begin{lemma}\label{point choice} If a compact  set $\bar U$ is in the causal past of an open set $V$, then there exists a finite collection of spatially separated points $\{u_i\}$ in $V$ such that $\bar U$ is in the chronological past of $\{u_i\}$.
 \end{lemma}
 
 \begin{proof} 
First, reduce to the case that $V$ is bounded. If $V$ is not bounded, we can choose an exhaustion of $V$ by compact subsets, and as $\bar U$ is compact, it is contained in the causal past of the interior of one of these compact subsets. So, without loss of generality, we can assume that $V$ is bounded.

 The future boundary of $V$ consists of the points  $p\in \bar V$ such that no point in $\bar V$ is in the causal future of $p$.  If $V$ is not bounded, the future boundary of $V$ may be empty, but as we are assuming that $V$ is bounded, the causal past of the future boundary of $V$ contains $V$, and therefore contains the causal past of $V$. The chronological past of each point in the future boundary of $V$ intersects $\bar U$ in an open subset, and together, these subsets form an open cover of $\bar U$. As $\bar U$ is compact, we can choose a finite sub cover, and obtain a finite set of points $\{p_i\}$ whose chronological past contains $\bar U$. As $p_i$ and $p_j$ are both in the future boundary of $V$, $p_i$ can not be in the causal future of $p_j$, so these points are spatially separated. Because we only have a finite set of points, there exist open neighborhoods $O_i$ of $p_i$ such that $O_i$ is spatially separated from $O_j$, and because $\bar U$ is compact, it is still covered by the chronological past of any set of points sufficiently close to our points $p_i$. Therefore, there exist  a choice of point $u_i\in O_i\cap V$ for all $i$ such that the chronological past of $\{u_i\}$ contains $\bar U$.
 
 \end{proof}
 
 \begin{lemma}\label{set separation} If $A$ and $B$ are  open sets, and $U$ is an open set whose causal past compactly contains  $A\cup B$, then there exist spatially separated bounded open subsets $U_A$ and $U_B$ of $U$ whose causal past compactly contains  $A$ and $B$ respectively. 
 \end{lemma}
 \begin{proof} First, use Lemma \ref{point choice} to choose a finite set of spatially separated points $\{u_1,\dotsc, u_n\}\subset U$  whose chronological past contains the closure of $A\cup B$. There exist spatially separated open neighborhoods $O_i$ of $u_i$ contained in $U$ so that any selection of points from each of these open neighborhoods has chronological past containing the closure of $A\cup B$. In each $O_i$, choose two non-empty spatially separated open subsets $U_{A,i}$ and $U_{B, i}$. Then $U_A:=\bigcup_i U_{A,i}$ is spatially separated from $U_B:=\bigcup_i U_{B,i}$, and the causal past of each compactly contains $A\cup B$, and therefore compactly contains both $A$ and $B$. 
 \end{proof}
 
 \begin{lemma}\label{separate representatives} Given any finite collection of totally classical projection-statements $P_i$, compactly supported in the causal past of an open set $U$, there exist spatially separated totally classical projection-statements $P_i'$  compactly supported within $U$ so that $P_i$ is compactly supported within the causal past of of $U_{P_i'}$  and so that  $P'_i$ is equivalent to $P_i$. 
 
 Moreover, given any such collection of projection-statements $P_i'$,  every logical expression $E'$ in the projection-statements $P_i'$ is also a totally classical projection-statement.
 
Furthermore, given any other spatially separated projection-statements $P_i''$ within $U$ satisfying the same conditions as $P_i'$, we have that $E'$ is equivalent to the corresponding logical expression $E''$ in the projection-statements $P_i''$. 
  \end{lemma}
 \begin{proof}
Using Lemma \ref{set separation}, there exist spatially separated bounded open subsets $U_{P_i'}\subset U$ whose causal past  compactly contains $U_{P_i}$. Then, as $P_i$ is totally classical, it is equivalent to a projection-statement $P_i'$ in $U_{P_i'}$, and this projection-statement is automatically totally classical, because any open set whose causal past compactly contains  $U_{P_i'}$   also compactly contains   $U_{P_i}$.

Now suppose that $P_i'$ and $P_i''$ are two collections of projection-statements satisfying the conditions required on $P_i'$. We need to  show that given any logical expression $E'$ in the $P_i'$, we have $E'\psi=E''\psi$, where $E''$ is the corresponding logical expression in the $P_i'$. If  the $P_i'$ were spatially separated from the $P_j''$, this would follow from Lemma \ref{record replacement}, so our proof will use the existence of more records satisfying appropriate spatial separation properties. Using Lemma \ref{set separation}, and the fact that $P_i$ is totally classical,  we can construct projection-statements $\tilde P_i'$ and $\tilde P_i''$ equivalent to $P_i$ such that  
\begin{itemize}
\item $U_{\tilde P_i'}$ is spatially separated from $U_{\tilde P_j''}$ for all $i\neq j$, and
\item $U_{\tilde P_i'}\subset U_{ P_i}$ and $U_{\tilde P_i''}\subset U_{P_i''}$.
  \end{itemize}
  
This implies that  $\tilde P_i'$ is spatially separated from both $P_j'$ and $\tilde P_j'$ for all $i\neq j$, so the hypothesis of Lemma \ref{record replacement} hold, with $\epsilon_i$ arbitrarily small. Accordingly, we get that for any logical expression $E'$ in the $P_i'$, we have that 
\[E'\psi=\tilde E'\psi\]
where $\tilde E'$ is the corresponding logical expression in the $\tilde P_j'$. Similarly, we get that 
\[\tilde E'\psi=\tilde E''\psi=E''\psi\]
so we get that $E'\psi=E''\psi$, as required.

It remains to show that $E'$ is a totally classical projection statement. If the causal past of an open set $O$ compactly contains  $ U_E'$, it also compactly contains $ U_{P_i}$ for all $P'_i$ appearing in the logical statement $E'$. Accordingly, Lemma \ref{set separation} implies that there exist spatially separated projection-statements $P_i''$ localised within $O$ and satisfying our requirements for $P'_i$. The above argument then implies that $E'\psi=E''\psi$ so $E'\psi$ is totally visible with $O$. So, $E'$ is a totally classical projection-statement.

 \end{proof}
 
 The following is an immediate corollary of Lemma \ref{separate representatives}.
 
 \begin{cor}\label{totally classical BA} There is a boolean algebra $\mathfrak B_U$ comprised of the set of equivalence classes $[P]$ of totally classical projection-statements P localised in regions compactly contained in $U$. The Boolean algebra structure is given by 
 \begin{itemize}
 \item $\neg [P]=[1-P]$, and
 \item  $[P]\wedge [Q]=[P'Q']$ where $P'$ and $Q'$ are spatially separated representatives of $[P]$ and $[Q]$. 
 \item  $[P]\vee [Q]=\neg(\neg [P])(\neg [Q])=[P'+Q'-P'Q']$,
 \item  and the addition in the corresponding Boolean ring is given by $[P]\oplus [Q]=[P'+Q'-2P'Q']$. 
 \end{itemize}
 
 Moreover, $[P]\mapsto P\psi$ defines an injective map $\phi:\mathfrak B_U\longrightarrow \mathcal H$ satisfying the property that 
 \[\phi([Q])=\phi( [Q]\wedge [P])+\phi( [Q]\wedge \neg [P])\]
 is an orthogonal decomposition of $\phi([Q])=Q\psi$.
   \end{cor}

\subsection{ $\mathcal H$--valued measure algebras }   
\label{ma section}

   The map
   \[\phi:\mathfrak B_U\longrightarrow \mathcal H\]
   defined by 
   $\phi([P]):=P\psi$
   gives $\mathfrak B_U$ the structure of a $(\mathcal H,\psi)$--valued measure algebra. This can be regarded as an intermediate notion between a projection-valued measure, and a real measure algebra associated to a probability space. We also discuss the relationship between  $\mathcal H$--valued measure algebras, statistical ensembles of states, and density matrices in Section \ref{density matrices}.

\begin{defn}\label{measure algebra} A $\mathcal H$--valued measure algebra is
\begin{itemize}\item a  Boolean algebra $\mathfrak B$, and 
\item an injective map $\phi:\mathfrak B\longrightarrow \mathcal H$.
\end{itemize}
such that, for all $a,b\in\mathfrak B$, 
\begin{equation}\label{additive}\langle \phi(a\wedge b)\vert \phi(a\wedge \neg b)\rangle=0 \ \ \ \text{ and } \ \ \phi(a)=\phi (a\wedge b)+\phi(a\wedge\neg b)\ .\end{equation}

If, in addition, $\phi(1)=\psi$, we say that $(\mathfrak B,\phi)$ is a $(\mathcal H,\psi)$--valued measure algebra.

\end{defn}

As $\mathfrak B$ is meant to represent a collection of `things one could say about Nature', we will refer to an element $b\in \mathfrak B$ as a statement. If we insist that statements about quantum systems take the form of a linear subspace, we can think of  $b\in\mathfrak B$ as asserting that the  quantum system is in the orthogonal complement of  the subspace spanned by $\phi(b')$ for all $b'\in \mathfrak B$ with $b'\wedge b=0$.

Note that equation (\ref{additive}) implies that $\mu(b) :=\norm{\phi(b)}^2$ is a  finitely additive measure. The definition of a real valued measure algebra $(\mathfrak B,\mu)$ is analogous, but more complicated as it usually requires $\mathfrak B$ to be $\sigma$--complete, and the measure $\mu$ to be countably additive instead of just finitely additive. We will show in Lemma \ref{completion} that any $\mathcal H$--valued measure algebra $(\mathfrak B,\phi)$ has a canonical completion to a $\sigma$--complete $\mathcal H$--valued measure algebra, $(\bar{\mathfrak B},\phi)$ and that $\phi$ is automatically countably additive. 

 We can construct a real valued measure algebra from a $\sigma$--complete $\mathcal H$--valued measure algebra by taking the measure $\mu(a)=\norm{\phi(a)}^2$. Given a projection valued measure, we can construct a $(\mathcal H,\psi)$--valued measure algebra by applying the projections to $\psi$, and taking the quotient by the ideal sent to $0$. Accordingly, we can construct a $(\mathcal H,\psi)$--valued measure algebra from a collection of commuting projections on $\mathcal H$.

 \begin{example}
Given a $\mathcal H$--valued measure algebra $(\mathfrak B,\phi)$, and some statement $O\in\mathfrak B$, there is an ideal $\mathfrak I_O\subset\mathfrak B$ comprised of statements $o\in\mathfrak B$ such that $o\wedge O=o$, and a filter\footnote{A filter $\mathfrak T\subset \mathfrak B$ is a proper subset of a boolean algebra $\mathfrak B$  closed under $\wedge$, and satisfying the property that $a\vee t\in\mathfrak T$ if $t\in\mathfrak T$. Filters often represent `things that are true'. An ideal $\mathfrak I\subset \mathfrak B$ is a nonempty subset closed under $\vee$ and satisfying the property that $a\wedge i\in\mathfrak I$ if $i\in \mathfrak I$.  } $\mathfrak T_O\subset\mathfrak B$ comprised of statements $b\in\mathfrak B$ such that $O\wedge b=O$. So, 
\[\mathfrak I_O=O\wedge \mathfrak B\ \ \ \ \text{ and }\ \ \ \ \mathfrak T_O=O\vee\mathfrak B\ .\]
When $O$ is a statement representing an observer's knowledge, $\mathfrak T_O$ represents statements that the observer can logically deduce, and $\mathfrak I_O$ represents statements that the observer should consider possible. 

  In fact, $(\mathfrak I_O,\phi)$ is a $(\mathcal H,\phi(O))$--valued measure algebra, and $(\mathfrak T_O,\phi-\phi(O))$ is a $(\mathcal H,\phi(1)-\phi(O))$--valued measure algebra.  The boolean algebra operations $\wedge$ and $\vee$ restrict from $\mathfrak B$ to $\mathfrak I_O$ and $\mathfrak T_O$, but the negation of a statement $o\in \mathfrak I_O$ is written as $(\neg o)\wedge O $ in terms of negation in $\mathfrak B$, and negation of $b\in\mathfrak T_O $  is $O\vee\neg b$. So, $O$ plays the role of the maximal element  $1$ in $\mathfrak I_O$ and the minimal element  $0$ in $\mathfrak T_O$. 
\end{example}

To see that any $\mathcal H$--valued measure algebra $(\mathfrak B,\phi)$ has a canonical completion, we need that the boolean  operations on $\mathfrak B$ are continuous. The following lemma shows that Boolean algebra operations are continuous, by showing that they are compatible with the metric on $\mathcal H$.

\begin{lemma} \label{bounded operations}Using the metric induced on $\mathfrak B$ by $\phi:\mathfrak B\longrightarrow \mathcal H$, the operation $\neg:\mathfrak B\longrightarrow \mathfrak B$ is an isometry, and the operations $\vee$  and $\wedge$ do not increase distances, so 
\[\norm{\phi(b\wedge c)-\phi(b'\wedge c)}\leq\norm{\phi(b)-\phi(b')}\geq\norm{\phi(b\vee c)-\phi(b'\vee c)}\]
\end{lemma}
\begin{proof} We have that $\phi(\neg b)+\phi(b)=\phi(1)$, so 
\[\phi(\neg b)-\phi(\neg b')=\phi(1)-\phi(b)+\phi(1)+\phi(b') = -(\phi(b)-\phi(b'))\]
and $\neg:\mathfrak B\longrightarrow \mathfrak B$ is an isometry.

We will now show that the map $b\mapsto b\wedge c$ does not increase distances.  Let $x=b'\wedge\neg b$ and $y=b\wedge\neg b'$ so we have that $x\vee b=y\vee b'$ and $x\wedge y=0=x\wedge b=y\wedge b'$. In particular, this implies that $\phi(x)+\phi(b)=\phi(y)+\phi(b')$ so, 
\[\phi(b)-\phi(b')=\phi(y)-\phi(x)\ .\]
 Then
\[\phi((b\vee x)\wedge c)=\phi((y\vee b')\wedge c)\]
and, as $b\wedge c$ is disjoint from $x\wedge c$, we get
\[\phi(b\wedge c)+\phi(x\wedge c)=\phi(y\wedge c)+\phi(b'\wedge c)\]
so
\[\phi(b\wedge c)-\phi(b'\wedge c)=\phi(y\wedge c)-\phi(x\wedge c)\]
As $x$ and $y$ are disjoint, $\phi(x)$ and $\phi(y)$ are orthogonal, and similarly, $\phi(y\wedge c)$ and $\phi(x\wedge c)$ are orthogonal. The fact that $\norm{\phi(y\wedge c)}\leq \norm{\phi(y)}$ then implies the required result:
\[\norm{\phi(b\wedge c)-\phi(b'\wedge c)}\leq \norm{\phi(y)-\phi(x)}=\norm{\phi(b')-\phi(c')}\]
The analogous result for the operation $\vee$ follows from the fact that $\neg$ is an isometry, and $a\vee b=\neg(\neg a\wedge \neg b)$.

\end{proof}

   \begin{lemma}\label{completion}Each $\mathcal H$--valued measure algebra $(\mathfrak B,\phi)$ has a unique completion $(\bar{\mathfrak B},\phi)$ to a $\mathcal H$--valued measure algebra such that the following diagram commutes
   \[\begin{tikzcd}\mathfrak B\dar\rar{\phi} & \mathcal H
   \\ \bar{\mathfrak B}\ar{ur}{\phi}\end{tikzcd}\]
   and such that $\phi(\bar{\mathfrak B})\subset\mathcal H$ is closed. 
   
    Moreover, $\bar{\mathfrak B}$ is $\sigma$--complete, with countable meets defined as
   \[\bigwedge_i b_i:=\lim_{n\to\infty} b_1\wedge\dotsb \wedge b_n\ .\]
and the countable join of $\{b_i\}$  defined as  
\[\bigvee_i b_i:=\lim_{n\to\infty} b_1\vee\dotsb \vee b_n\ .\]
With this structure, $\phi$ is countably additive, in the sense that if $b_i\wedge b_j=0$ for all $i,j\in\mathbb N$, 
\[\phi\left(\bigvee_i b_i\right)=\sum_i\phi(b_i):=\lim_{n\to\infty} \sum_{i=1}^n\phi(b_i)\ .\]

    \end{lemma}
\begin{proof} Let $\bar{\mathfrak B}$ be the completion of $\mathfrak B$ as a metric space, using the induced metric from $\mathcal H$.
Lemma \ref{bounded operations} implies that the Boolean algebra operations $\neg, \wedge, \vee$ are continuous with respect to our metric, so extend to the completion $\bar {\mathfrak B}$. As the axioms of a Boolean algebra or ring are written as equations in these continuous operations, these axioms are also satisfied on the completion. So, $\bar {\mathfrak B}$ is a Boolean algebra. The map $\phi:\mathfrak B\longrightarrow\mathcal H$ is continuous, so also extends uniquely to a continuous map $\phi:\bar{\mathfrak B}\longrightarrow \mathcal H$. As equation (\ref{additive}) is comprised of continuous operations, it is also satisfied on $\bar{\mathfrak B}$, so $(\bar{\mathfrak B},\phi)$ is a $\mathcal H$--valued measure algebra, whose image in $\mathcal H$ is the closure of $\phi(\mathfrak B)$. It is also the unique such algebra, because  the Boolean algebra operations and $\phi$ are automatically continuous, so must coincide with the extensions of the analogous operations on $\mathfrak B\subset\bar{\mathfrak B}$.

It remains to show that $\bar{\mathfrak B}$ is $\sigma$--complete and that $\phi$ is countably additive. Given a countable collection of elements $b_i$ of $\bar {\mathfrak B}$, the sequence of finite joins
 \[ c_n:=b_1\wedge\dotsb \wedge b_n\]
 is a Cauchy sequence. To see this, note that 
 \[\phi(c_{n})=\phi(c_{n+k})+\phi(c_{n}\wedge\neg c_{n+k})\]
 is an orthogonal decomposition, so  
 \[\norm{\phi(c_{n})-\phi(c_{n+k})}=\sqrt{\norm{\phi(c_n)}^2-\norm{\phi(c_{n+k})}^2}\ .\]
 We know that $\norm{\phi(c_n)}^2$ is a Cauchy sequence because it is bounded below and monotonically decreasing, so we can use the inequality above to conclude that  $c_n$ is also a Cauchy sequence.  Accordingly, $c_n$ has a unique limit in $\bar {\mathfrak B}$, so we can define the countably infinite meet of $\{b_i\}$ as 
\[\bigwedge_i b_i:=\lim_{n\to\infty} b_1\wedge\dotsb \wedge b_n \ .\]
This is the infimum of $\{b_i\}$. In particular we have  that $b_j \bigwedge_ib_i=\bigwedge_i b_i $ because this equation holds with $\bigwedge_i b_i$ approximated by $c_n$ with $n>j$. Moreover, if $b_i\wedge a=a$ for all $i$, we have that $a\bigwedge_i b_i=a$, because $a\wedge c_n=a$ for all $n$. These two properties identify $\bigwedge_ib_i$ as the infimum of $\{b_i\}$. 
Similarly, we can define the supremum of $\{b_i\}$ as
\[\bigvee_i b_i:=\lim_{n\to\infty} b_1\vee\dotsb \vee b_n\ .\]

Then, as $\phi$ is continuous we have
\[\phi\left(\bigvee_i b_i\right)=\lim_{n\to\infty} \phi(b_1\vee\dotsb \vee b_n)\ .\]
It follows that if $b_i\wedge b_j=0$, we have
\[\phi\left(\bigvee_i b_i\right)=\lim_{n\to\infty} \sum_{i=1}^n\phi(b_i)\ .\]
so $\phi$ is countably additive.
\end{proof}

   \begin{remark} If the causal past of $V$ contains $U$, then \[(\mathfrak B_U,\phi)\subset (\mathfrak B_V,\phi)\ .\] 
   The completion of $(\mathfrak B_U,\phi)$ to a $\sigma$--closed $\mathcal H$--valued measure algebra $(\bar{\mathfrak B}_U,\phi)$ is functorial, so whenever the causal past of $V$ contains $U$, we still have
\[(\bar{\mathfrak B}_U,\phi)\subset (\bar{\mathfrak B}_V,\phi)\ .\]
    This is analogous to the strongest condition we could expect on $\sigma$-algebras $\mathfrak F_U$ encoding classical information. 
   \end{remark}

Let us replace our $\mathcal H$--valued measure algebras $\bar{\mathfrak B}_U$ with $\sigma$-algebras $\mathfrak F_U$ on some space $X$. Let $\mathfrak B$ be the Boolean algebra of totally classical projection statements, that is, $\mathfrak B=\mathfrak B_U$ for  $U$  the entire space-time. Then define
\[X:=\hom(\bar {\mathfrak B},\{0,1\})\]
as the set of Boolean algebra homomorphisms from $\bar{\mathfrak B}$ to the $2$-element Boolean algebra. Each element of $X$ corresponds to a logically consistent assignment of `true' or `false' to each statement in $\bar{\mathfrak B}$ --- the set of `true' statements sent to 1 is called an ultrafilter, and the set of `false' statements sent to $0$ is called a maximal ideal. For any statement $b\in\bar{\mathfrak B}$, define
\[S_{b}\subset X\]
to be the set of homomorphisms sending $b$ to $1$. Let $\mathfrak F$ be the $\sigma$--algebra generated by $S_{b}$ for all $b\in \bar{\mathfrak B}$, and let $\mathfrak F_U\subset \mathfrak F$ be the $\sigma$--algebra generated by $S_{b}$ for all $b\in\bar{\mathfrak B}_U$.

 \begin{lemma}\label{space from algebra}
There exists a unique measure $\mu$ on $(X,\mathfrak F)$ such that $\mu(S_{b})=\norm{\phi(b)}^2$. Moreover, with this measure on $X$, $\bar{\mathfrak B}_U$ is the quotient of $\mathfrak F_U$ by the ideal of measure zero subsets.  We also have that, $(X,\{\mathfrak F_U\},\mu)$ satisfies our expectations for classical information: if the causal past of $V$ contains $U$, then
\[\mathfrak F_U\subset \mathfrak F_V\]
 \end{lemma}
 \begin{proof}
 The Stone isomorphism represents an arbitrary Boolean algebra, such as $\bar{\mathfrak B}$, as the Boolean algebra of closed and open subsets of the set of Boolean algebra isomormorphisms $\bar{\mathfrak B}\longrightarrow \{0,1\}$. In particular, sending $b$ to $S_{b}$ is an isomorphism between  $\bar{\mathfrak B}$ and the Boolean algebra of closed and open subsets of $X$.
    The Loomis--Sikoski representation theorem tells us that $\bar {\mathfrak B}$ is the quotient of  $\mathfrak F$ by some sub $\sigma$-algebra $\mathfrak N\subset \mathfrak F$; see \cite{Sikorski} or the discussion of this theorem on Terry Tao's blog, \cite{TaoLS}. We can therefore define $\mu$ on $\mathfrak F$ by pulling back the measure $\mu(b):=\norm{\phi(b)}^2$ from $\bar{\mathfrak B}=\mathfrak F/\mathfrak N$. This measure satisfies $\mu(S_{b})=\norm{\phi(b)}^2$, and the fact that $\phi:\bar{\mathfrak B}\longrightarrow \mathcal H$ is injective implies that $\mathfrak N$ is the ideal of measure $0$ subsets. Similarly, $\bar {\mathcal B}_U$ is the quotient of $\mathfrak F_U$ by the ideal of measure $0$ subsets in $\mathfrak F_U$.
    
Moreover,   if the causal past of $V$ contains $U$ we have that $\mathfrak B_U\subset \mathfrak B_V$, so $\bar{\mathfrak B}_U\subset\bar{\mathfrak B}_V$, and therefore $\mathfrak F_U\subset\mathfrak F_V$.

\end{proof}

\subsection{Relative verses absolute classical reality.}

We have now constructed a probabilistic model $(X,\{\mathfrak F_U\},\mu)$ for classical information. In Lemma \ref{separate representatives} and Corolllary \ref{totally classical BA} we showed that totally classical projection-statements form $(\mathcal H,\psi)$--valued measure algebras $(\mathfrak B_U,\phi)$. In Lemma \ref{completion}, we completed these measure algebras to obtain $(\bar{\mathfrak B}_U,\phi)$ and then in Lemma \ref{space from algebra} constructed $(X,\{\mathfrak F_U\},\mu)$ so that $(\bar{\mathfrak B}_U,\mu)$ is the quotient of $(\mathfrak F_U,\mu)$ by the ideal of measure zero subsets.

In section \ref{approximate}, we will approximate this classical model, but before going on, we emphasise what we did not construct, and will not approximate. Many intuitive conceptions of classical information include a point $\mathfrak r\in X$ to represent `what is really true'. In terms of the boolean algebra $\bar{\mathfrak B}$ such a point $\mathfrak r\in X$ corresponds to an ultrafilter $\mathfrak T\subset \bar{\mathfrak B}$ comprised of the set of all statements that are `really true'. We did not construct this, and will not attempt to approximate such a thing, because this becomes deeply problematic\footnote{Constructing $\mathfrak r\in X$ would require some kind of dynamic wavefunction-collapse theory, like \cite{Dynamic_collapse}. It is hard to imagine a version of this compatible with relativity.  } in the quantum setting.   

Without a choice of absolute truth $\mathfrak T\subset \bar{\mathfrak B}$, we are left with only a relative notion of truth, such as advocated by Rovelli in \cite{Rovelli}, or Everett in \cite{Everett}. Such a relative notion of truth is perfectly sufficient to account for what we observe, and how we make probabilistic predictions.

Recall how we represent the information known by an observer within $U$. This information is represented by a statement $O\in\mathfrak B_U$. This observer can deduce the statements in the filter $\mathfrak T_O =\bar{\mathfrak B}\vee O$, and  considers statements in the ideal $\mathfrak I_O=\bar{\mathfrak B}\wedge O$ to be possibilities. This observer assigns probabilities determined by the restriction of the probability measure $\mu$ to $\mathfrak I_O$, rescaled by $\frac 1{\mu(O)}$  --- this is the usual Bayesian update rule. 

Intuitively and linguistically, we are accustomed to assign an observer a unique identity throughout time. Such a unique identity corresponds to a statement $O_t\in \mathfrak B$ for each time. We expect that, at a later time $t'>t$,  the observer knows their previous state, so $O_t\in \mathfrak T_{O_{t'}}$. Equivalently, $O_{t'}\in\mathfrak I_{O_t}$ for $t'>t$, so, at time $t$, our observer considers  the future state $O_{t'}$ as a possibility.   Here, we get to a nonintuitive point: There will generally be other, equally valid, statements $O'_{t'}\in\mathfrak I_{O_t}$ that could describe the future state of our observer $O_t$. Moreover, these statements can be incompatible, so $O'_{t'}\wedge O_{t'}=0$. For example, these two statements might correspond to observing different outcomes of a quantum measurement. The same is not true for past states of our observer.  If $O_s$ and $O'_s$ both describe a past state of our observer, these must be compatible, because $O_s\wedge O_s'\wedge O_t=O_t$. 

Because our observer does not know of any incompatible past states, it is a reasonable leap for them to not believe in incompatible future states. After all, our observer will never observe such incompatible future states. Our observer's knowledge, represented by the filter $\mathfrak T_{O_t}\subset \mathfrak B$, also increases over time.  Our observer can then be forgiven for making a logical leap, and  believing in a maximal filter $\mathfrak T\subset \mathfrak B$ representing all things that are `really true' --- they then think of $\mathfrak T_{O_t}\subset \mathfrak T$ as the collection of really true statements that they can deduce from their current knowledge.    
 
Although belief in absolute truth is compatible with classical observations, let us consider some consequences of a choice of ultrafilter $\mathfrak T\subset\bar {\mathfrak B}$ representing absolute truth. To each spacetime region $U$, we then have an ultrafilter 
\[\mathfrak T_U=\mathfrak T\cap  \bar{\mathfrak B}_U\] representing `things that are really true and accessible within $U$'. Now suppose that our spacetime is globally hyperbolic, so foliated by Cauchy surfaces $S_t$.  With this notion of time, we can consider $\mathfrak T_t:=\mathfrak T_{S_t}$ as the `things that are really true at time $t$'. For future times $t'>t$, the causal past of $S_{t'}$ contains $S_t$, so $\mathfrak B_{S_t}\subset \mathfrak B_{S_{t'}}$, and therefore $\mathfrak T_t\subset\mathfrak T_{t'}$. This has the reassuring interpretation that `things that were true stay true'. Similarly, `things true in a spacetime region $U$ remain true in any region whose causal past contains $U$'. However,  assuming that $\mathfrak T_{t}$ is a proper subset of $\mathfrak T_{t'}$, we have that our model is non-deterministic: Things that are true at time $t$ are not sufficient to determine everything that is true at time $t'>t$.

Occam's razor might be used to argue for the existence of a unique classical reality, as it appears more complicated  to have `multiple valid classical realities'.   If, however,  one starts with the structure of the measure algebras $(\mathfrak B_U,\phi)$, a choice of ultrafilter $\mathfrak T$ is a lot of extra information, with no physical law available to help with the choices involved. We can use the axiom of choice to know that such an ultrafilter exists, but this in no way helps justify the belief in a unique canonical ultrafilter.  Intuitively, most people believe in $\mathfrak T_t$, representing what is `really true right now', but  fewer believe in predestination, which corresponds to a choice of extension of $\mathfrak T_t$ to $\mathfrak T_{t'}$ for all future times $t'>t$. Quantum mechanics has underlined that there are no sensible rules for determining $\mathfrak T_{t'}$ given $\mathfrak T_{t}$, so predestination is an unjustified belief. Similarly, if $U$ is some open region on the Cauchy surface $S_t$, there are no sensible rules for determining $\mathfrak T_{t}$ from $\mathfrak T_U$, so a belief in  such an extension $\mathfrak T_t\supset \mathfrak T_U$ is also an unjustified belief. The point of view of this paper is that, from the perspective of the observer $O$, any extension of the filter $\mathfrak T_O$ is an unjustified belief, akin to a belief in predestination.

In the classical setting, believing in an unique ultrafilter $\mathfrak T$ representing classical truth is harmless, if unnecessary. In the quantum setting, an analogous belief becomes problematic, because it requires a physical wavefunction collapse mechanism. Such a belief is also not compatible with the version of probability given in this paper.

\section{Approximations to classical information}
\label{approximate}

In Section \ref{precise}, we constructed a model  $(X,\{\mathfrak F_U\},\mu)$ for classical information. It is worth pausing to consider what aspects of this model are unrealistic. 

Realistically, we only expect classical projection-statements to exist, and these are just an approximation of totally classical projection-statements. So,  $(X,\{\mathfrak F_U\},\mu)$ is a mathematically convenient fantasy, emerging from the actual situation in the sense that many predictions we can make using $(X,\{\mathfrak F_U\},\mu)$ are approximately true. We don't expect classical projection-statements localised in $U$ to actually form a Boolean algebra, let alone a $\sigma$--closed Boolean algebra. What we actually expect is that small logical operations on classical projection-statements are themselves approximated by classical projection-statements. We also expect that the probabilities we assign to classical projection-statements only approximately obey the axioms of a probability measure. The goal of this section is to make these expectations more precise.

\subsection{Classical $\mathcal H$--valued measure algebras}

In Corollary \ref{totally classical BA}, each statement in the $(\mathcal H,\psi)$--valued measure algebra $(\mathfrak B_U,\phi)$ was an equivalence class of totally classical projection-statements, totally visible in the spacetime region $U$. This is not realistic, but given a $\mathcal H$--valued measure algebra $(\mathfrak B,\phi)$ we still can  ask the question `where in spacetime can the statement $b\in \mathfrak B$ be accessed?'

\begin{defn}\label{statement visibility}Define the visibility of a statement $b\in\mathfrak B$ in $U$ to be 
\[\text{Vis}_U(b):=\begin{cases}\sup_{P}\ln\norm{\phi(b)}-\ln\norm{P\psi-\phi(b)} &\text{ if }b\neq 0
\\ 0 &\text{ if }b=0\end{cases}\]

where the supremum is taken over projection-statements $P$ supported in $U$.
\end{defn}
The visibility of a statement $b$ in a region $U$ takes values in $[0,\infty]$, with higher numbers indicating more visibility. It is highly relevant where a statement is visible. Not only does this determine where the information is accessible, it is also essential for determining when logical combinations of classical projection-statements should be classical, using Lemmas \ref{replacement} and \ref{record replacement}. For a statement $b\in \mathfrak B$ to be classical, we want $b$ to be visible in many spatially separated regions in spacetime. We also need an extra condition so that logical operations in $\mathfrak B$ reflect operations with classical projection-statements.

\begin{defn}\label{classical statement def} A statement $b$ in  $(\mathfrak B,\phi)$ is classical if there exists a classical projection-statement $P$ such that,
\[\norm{P\psi-\phi(b)}\ll \phi(b)\] 
and such that, 
 for all $a\in\mathfrak B$, 
\[\norm{\phi(b\wedge a)-P\phi(a)}\ll \norm{\phi(b)}\]
\end{defn}
So, a classical statement $b\in\mathfrak B$ is one that is closely approximated by a classical projection-statement $P$, in the sense that $\phi(b)\approx P\psi$ and  $\phi(b\wedge a)\approx P\phi(a)$, and the approximation is close compared to the size of $\norm{\phi(b)}$.

\begin{defn}\label{new classical measure algebra def} A $\mathcal H$--valued measure algebra $(\mathfrak B,\phi)$ is classical, if for some $\epsilon\ll \norm{\phi(1)}$, the set of classical statements in $\mathfrak B$ is $\epsilon$--dense. So, each statement in $b$ is within $\epsilon$ of a classical statement. 
\end{defn}

 Lemma \ref{bounded operations} implies that, given a classical measure algebra $(\mathfrak B,\phi)$ in which classical statements are $\epsilon$--dense, all statements $b\in\mathfrak B$ with $\norm{\phi(b)}\gg \epsilon$ are classical.

\begin{example} Suppose we want to model a continuous measurement with a projection--valued measure on $\mathbb R^n$. By applying this projection-valued measure to some state in $\mathcal H$, we can obtain a measure $\mu$ on $\mathbb R^n$ and a $\mathcal H$--valued measure algebra $(\mathcal B,\phi)$  We would like that some outcomes of this measurement are classical statements that are visible in many separate regions of spacetime.  Which measurement outcomes, (represented by subsets $S\subset\mathbb R^n$), could we hope to be classical? We need the measure of $S$ to be sufficiently large, but this is not all that is needed, as we can expect $S$ to be less classical if it is hard to distinguish from its complement.
\end{example}

Given a classical $\mathcal H$--valued measure algebra $(\mathfrak B,\phi)$, if an observer is described by a classical statement $O\in \mathfrak B$, the ideal $\mathfrak I_O\subset\mathfrak B$ with maximal element $O$ is also a classical $\mathcal H$--valued measure algebra, so long as $ \norm{\phi(O)}$ is sufficiently large. From the perspective of this observer, this classical $\mathcal H$--valued measure algebra $(\mathfrak I_O,\phi)$ represents the statements in  $(\mathfrak B,\phi)$ that are possibilities. Associated to $O$, we also have a filter $\mathfrak T_O\subset\mathfrak B$, comprised of all statements $b\in \mathfrak B$ such that $b\wedge O=O$. These are the statements that the observer regards as true.

In section \ref{cma construction}, we will explain that classical $(\mathcal H,\psi)$--valued measure algebras can be constructed using spatially separated records of classical projection-statements, and show that the resulting measure is approximately independent of the records chosen; see Heuristics \ref{classical measure algebra} and \ref{algebra record replacement}. Thus, we can expect that probabilistic reasoning with classical projection-statements is approximately valid.

\subsection{Bell's theorem }

\label{bell section}

Bell's theorem provides an obstruction to a probabilistic model for quantum information. In this section, we clarify how Bell's theorem applies in our setting.

For any finite,  spatially separated collection, $I$, of  projection-statements, denote by $\mathfrak B_I$ the Boolean algebra generated by these projection-statements, with $P\wedge Q:=PQ$ and $\neg P:=1-P$. For this construction to work, we need $P$ and $Q$ to commute, so that  $PQ$ is a projection, and $P\wedge Q=Q\wedge P$. Let $\phi:\mathfrak B_I\longrightarrow\mathcal H$ be the  $(\mathcal H,\psi)$--valued measure sending $P$ to $P\psi$, and let $\mu:\mathfrak B\longrightarrow [0,1]$ be the corresponding probability measure defined by $\mu(P)=\norm{P\psi}^2$. Under the assumption that $P\psi\neq 0$ for $P\neq 0$, we have that $(\mathfrak B_I,\mu)$ is a measure algebra and $(\mathfrak B_I,\phi)$ is a $(\mathcal H,\psi)$--valued measure algebra.

If $I\subset J$ we have a natural inclusion of $\mathcal H$--valued measure algebras
\[\iota_{I\subset J}: \mathfrak B_I\longrightarrow \mathfrak B_J\]
and consequently $(\mathfrak B_I,\phi)$ is a coarse-graining of $(\mathfrak B_J,\phi)$ and $(\mathfrak B_J,\phi)$ is a fine-graining of $(\mathfrak B_I,\phi)$. 

Coming from a classical probabilistic world view,   we could hope  that there is some universal $\mathcal H$--valued measure algebra $(\mathfrak B,\phi)$ so that each of our measure algebras $(\mathfrak B_I,\phi)$ is canonically a coarse-graining  of $(\mathfrak B,\phi)$, compatible with  our natural inclusions  $\iota_{I\subset J}$  in the sense that the following diagram commutes.
\[\begin{tikzcd}(\mathfrak B_I,\phi)\rar{\iota_{I\subset J}}\ar{dr} & (\mathfrak B_J,\phi)\dar
\\ & (\mathfrak B,\phi)\end{tikzcd}\]
Given such a $(\mathfrak B,\phi)$, we could then construct a probabilistic model for quantum information, as in Lemma \ref{space from algebra}. In the language of category theory, we hope for a cocone, $(\mathfrak B,\phi)$,  of the diagram comprised of measure algebras $(\mathfrak B_I,\phi)$  and inclusions $\iota_{I\subset J}$.
 A less infinite version of this hope is that any finite sub-diagram has a cocone, so we could hope that any finite collection of our measure algebras has a common fine-graining compatible with the inclusions $\iota_{I\subset J}$.  In the general setting of quantum mechanics, this hope is crushed by Bell's theorem, even for real-valued measure algebras.  Given projection-statements $A_0$ and $A_1$, spatially separated from projection-statements $B_0$ and $B_1$, there does not in general exist a common fine-graining $(\mathfrak B,\phi)$ so that the following diagram commutes. Moreover, even if we  only require the maps to be approximately measure preserving, such a diagram still does not exist.
 \begin{equation}\label{bell cd}\begin{tikzcd} {\mathfrak B}_{A_0} \dar\ar{dr}& \ar{dl}\ar{dr} \mathfrak B_{B_0} &\dar\ar{dr} \mathfrak B_{A_1} &  {\mathfrak B}_{B_1}\ar{dll}\dar
 \\ {\mathfrak B}_{A_0,B_0}\ar[dotted]{dr} & {\mathfrak B}_{A_0,B_1} \dar[dotted] & {\mathfrak B}_{A_1,B_0}\ar[dotted]{dl}& {\mathfrak B}_{A_1,B_1}\ar[dotted]{dll}
 \\ & {\mathfrak B} \end{tikzcd}\end{equation}
Let us recall Bell's argument in this context. Assuming such a diagram existed, we get a corresponding diagram of probability spaces $X_{\mathfrak B_*}$ with arrows reversed, where $X_{\mathfrak B_*}$ indicates the set of homomorphisms $\mathfrak B_*\longrightarrow\{0,1\}$ as in Lemma \ref{space from algebra}. Given $\{\pm 1\}$--valued functions $a_i$ on $X_{\mathfrak B_{A_i}}$,  and  $b_i$ on $X_{\mathfrak B_{B_i}}$, we have corresponding functions on $X_{\mathfrak B}$, and the function $(a_0+a_1)b_0 +(a_0-a_1)b_1$ on $X_{\mathfrak B}$ is bounded by $2$, and hence has expectation value bounded by $2$. On the other hand, this is the sum of the expectation values of $a_0b_0$, $a_1b_0$, $a_0b_1$ and $-a_1b_1$, and these expectation values can be calculated separately on our given probability spaces $X_{\mathfrak B_{A_i,B_j}}$. The problem observed by Bell is the following: using entangled spin states and non-commuting projections, it is easy to construct situations where the sum of these expectations is $2\sqrt 2$, so no fine-graining exists in such situations. Moreover, even if the maps in diagram (\ref{bell cd}) were approximately measure preserving maps of boolean algebras, we would still get an bound on our expectation value by something close to $2$, so Bell's theorem also rules out this situation.  So, Bell's theorem tells us that we have no hope of a classical probabilistic model for quantum information.

Instead of looking for a classical probabilistic model for \emph{all} quantum information, we are seeking  approximate probabilistic models $(\mathfrak B_I,\phi)$ built only out of classical projection-statements $P\in I$. We will show that such models are classical $(\mathcal H,\psi)$--valued measure algebras, and that small collections of such classical measure algebras  have an approximate common fine-graining. The reason that this will not violate Bell's theorem is as follows:  The  expectation values of $a_ib_j$ can not be directly measured but can be approximated by taking several repeat experiments with two spatially separated experimenters:  Alice choosing from $A_0$ or $A_1$ to measure $a_0$ or $a_1$, and Bob choosing from $B_i$ to measure $b_i$. When Alice chooses to measure $a_0$, this involves a physical interaction that makes  $A_0$ a classical projection-statement\footnote{More accurately, Alice's choice of measurement determines an experimental setup described by a projection-statement $E_{A_i}$, and then $A_iE_{A_i}$ becomes a classical projection statement, but $A_0E_{A_1}$ and $A_1E_{A_0}$ will not be classical projection statements. Then each instance of $A_i$ in diagram (\ref{bell cd}) has to be replaced by the two projection statements $E_{A_i}$ and $A_iE_{A_i}$, and similarly each instance of $B_i$ must be replaced by $\{B_iE_{B_i}, E_{B_i}\}$. Then, there is no problem with this modified diagram (\ref{bell cd}) existing, because now, in our common fine graining $\mathfrak B_{\{A_iE_{A_i}, E_{A_i},B_iE_{B_i},E_{B_i}\}}$, we can only calculate expectations of of $a_i$ and $b_i$ conditioned on the relevant choices of experimental setups.   }, but in this case, there is no reason for $A_1$ to be a classical projection-statement, so we do not require our fine-graining to be compatible with $\mathfrak B_{A_1}$, and the reasoning of Bell's theorem does not apply.

\subsection{Construction of classical $\mathcal H$--valued measure algebras}
\label{cma construction}

Choose a space-like surface $S$ and consider the collection of records, $P$ on $S$, of classical projection-statements. We consider these as giving the `classical information available on $S$'. The key extra thing we have now is that each such record $P$ is approximated by many spatially separated records $P'$, so 
\[P\psi\approx P'\psi\ .\]

Use the notation $\mathfrak B_I$ for the Boolean algebra generated by a spatially separated collection, $I$, of records of classical projection-statements. This is as in Section \ref{bell section}, with the key difference that each of these projection-statements will also be visible in many spatially separated regions. If we have a different set $I'$ of records of the same collection of classical projection statements, the bijection between $I$ and $I'$ induces an isomorphism  of Boolean algebras 
\[f: \mathfrak B_I\longrightarrow \mathfrak B_{I'}\ .\] 

\begin{heuristic}\label{separated record replacement} Given a collection of classical projection-statements in the past of a space-like hypersurface $S$, and two collections, $I$ and $I'$ of spatially separated records of these classical projection statements on $S$, the isomorphism $f:\mathfrak B_I\longrightarrow \mathfrak B_I'$  approximately preserves the $\mathcal H$--valued measure algebra structure, so the following diagram approximately commutes.
\[\begin{tikzcd}\mathfrak B_I \rar{f}\dar{\phi} & \mathfrak B_{I'}\ar{dl}{\phi}
\\ \mathcal H \end{tikzcd}\]
\end{heuristic}
As a consequence,  the probabilities we assign to logical statements in records of classical projection-statements are approximately independent of what records we choose. 

To justify Heuristic \ref{separated record replacement}, we want that $\phi(f(P))\approx \phi(P)$ for $P\in \mathfrak B_I$. As $I$ and $I'$ are records of a single collection of classical projection-statements,  we have that this is true for $P\in I$, but it is by no means obvious that this still holds for logical statements $E$ in the projection-statements in $I$.
Suppose that the projection-statements in $I$ are spatially separated from the projection-statements in $I'$. Then, Lemma \ref{record replacement} implies that 
\begin{equation}\phi(E)=E\psi\approx f(E)\psi=\phi(f(E))\end{equation}
so $f:(\mathfrak B_I,\phi)\longrightarrow (\mathfrak B_{I'}, \phi)$ is an approximately measure preserving map, and we can conclude Heuristic \ref{separated record replacement} in this case.  Even if the projection-statements in $I$ are not spatially separated from the projection-statements in $I'$, Lemma \ref{record replacement} still implies that $E\psi\approx f(E)\psi$ so long as the projection-statements in $I$ are visible in some region spatially separated from all projections in $I$ and $I'$. Note however that the error in this approximation, $\norm{E\psi-f(E)\psi}$ is significantly smaller when $E$ can be written as a small logical expression in the projection-statements $P\in I$. So, we have the following caveat:
\begin{heuristic}\label{caveat}Even though the Boolean isomorphism $f:\mathfrak B_I\longrightarrow \mathfrak B_{I'}$ is an approximately measure preserving map, the approximation is better for small logical combinations of the projection-statements $P$ in $I$. In particular, we can expect the size of the error for $a\wedge b$ or $a\vee b$ to be bounded by the sum of the error  for $a$ and for $b$.
\[\norm{\phi(a\wedge b)-\phi(f(a\wedge b))}\leq \norm{\phi(a)-\phi(f(a))}+\norm{\phi(b)-\phi(f(b))}\]
\end{heuristic}

Heuristic \ref{separated record replacement} implies that, for assigning probabilities to logical combinations of classical projection-statements, it should not matter what spatially separated records we chose. At the expense of obtaining a worse approximation, Heuristic \ref{separated record replacement} can be upgraded to remove the condition that both collections of records are on the same space-like hypersurface. In particular, we can choose a sequence of spatially separated records $I_1=I$ through $I_n=I'$, such that the projections in  $I_n$ and  $I_{n+1}$ are spatially separated. So long as $n$ is not too large, we arrive at the following.

\begin{heuristic}\label{algebra record replacement} Given two collections $I$ and $I'$ of spatially separated records of a given collection of classical projection-statements, the boolean isomorphism $f:\mathfrak B_I\longrightarrow \mathfrak B_{I'}$ is  approximately measure preserving.
\end{heuristic}

As a corollary, if $I$ consists of spatially separated records of some collection of classical projection-statements, and $J$ includes records of these same projection-statements, (and possibly records of some other projection statements), we get a canonical inclusion of boolean algebras
\[\mathfrak B_I\longrightarrow \mathfrak B_J\]
that is approximately measure preserving. So, $(\mathfrak B_J,\phi)$ is an approximate fine-graining of $(\mathfrak B_I,\phi)$. Moreover, we expect that sufficiently small collections of such $(\mathfrak B_{I_i},\phi)$  have a common approximate fine-graining.   This is precisely what Bell's theorem rules out when we use arbitrary projection-statements in place of classical projection-statements.

\

With Heuristic  \ref{algebra record replacement}, we can  argue that $(\mathfrak B_I,\phi)$ should be  a classical $(\mathcal H,\psi)$--valued measure algebra, in analogy to Lemma \ref{separate representatives}. In particular, each record $P\in I$ should have many spatially separated records $P'$ in the causal future of $P$, because $P$ is a record of a classical projection statement. So we can expect to find some spatially separated collection of such records $I'$. With Heuristic \ref{algebra record replacement}, we can expect that $f:\mathfrak B_I\longrightarrow\mathfrak B_{I'}$ is approximately measure preserving, so for each $E\in \mathfrak B_I$, we expect that
\[E\psi=\phi(E)\approx\phi(f(E))=f(E)\psi\]
So long as $\norm{\phi(E)}$ is large enough compared to the error in the above approximation, we then get that $f(E)$ is a record of $E$, and can conclude that $E$ itself is a classical projection-statement. We therefore arrive at the following.
\begin{heuristic} \label{classical measure algebra} Suppose that $I$ is a small collection of spatially separated records of classical projection-statements. Then, we expect $(\mathfrak B_I,\phi)$ to be a classical $(\mathcal H,\psi)$--valued measure algebra. 

\end{heuristic}

 \section{Decoherence leading to classical projection-statements}\label{Decoherence}

The main thesis of this paper the following: our perception of classical reality within a quantum system is described by classical projection-statements. So far, we have not addressed the question of how classical projection-statements could possibly arise within a quantum system. It is certainly not true that nontrivial classical projection-statements exist for any choice of quantum system $(\mathcal H,\psi)$, even given a notion of projection-statements localised in regions of spacetime.  In this section, we give simplified examples demonstrating how classical projection-statements can arise.\footnote{A more complete discussion of the topic would go significantly further by demonstrating that classical projection-statements exist in realistic quantum models of our universe, and further demonstrating that our perceptions are sufficiently represented by these classical projection-statements. This would be a significantly more challenging project. } We also develop some heuristics limiting how visible we can expect projection-statements to be. These limits to visibility are important for our theory of natural probability. 

The decoherence program has gone a long way towards explaining aspects of measurement in quantum physics --- for surveys see \cite{survey,evidence,Zurekreview}. Decoherence means different things to different people. It can refer to:
\begin{enumerate}
\item loss of interference patterns exhibited a quantum subsystem, (important to avoid when designing a quantum computer),
\item  irretrievable loss of information from a quantum subsystem via entanglement with an inaccessible environment, or
\item    information from a quantum subsystem being recorded in fragments of the environment, or Quantum Darwinism, \cite{Zurekreview}.
\end{enumerate}
These aspects of decoherence provide satisfactory resolution to many apparently paradoxical aspects of quantum measurement, without resorting to wavefunction collapse.
 \begin{example}\label{2 slit}Consider a double slit experiment. Such an experiment sends photons  through two slits, producing an interference pattern on a screen, even when photons are sent through one at a time. This interference pattern is destroyed when we measure which slit the photons pass through. This experiment contradicts the idea that measurement is simply an update of information. This, apparently  paradoxical, situation is easy to explain without wavefunction collapse. The unmeasured photons reaching the screen can be represented by a linear combination $\frac 1{\sqrt 2}(\ket{\text{left}}+\ket{\text{right}})$ of wavefunctions representing a photon passing through the left or the right slit. We can include the measurement apparatus in our quantum system, then model our situation by a unitary transformation (or premeasurement) as follows
 \[\frac 1{\sqrt 2}(\ket{\text{left}}+\ket{\text{right}})\otimes\ket{\text{apparatus}}\mapsto \frac 1{\sqrt 2}(
 \ket{\text{left}}\otimes\ket{l}+\ket{\text{right}}\otimes\ket{r}) \]
After this unitary transformation, the interference pattern on the screen disappears, because of how the photon is entangled with the apparatus. \end{example}

\begin{example} \label{cat}Decoherence can readily explain a lack quantum effects between  `Schr\"odinger cat states', comprised of a linear combination of wavefunctions describing macroscopically different arrangements of matter, such as an unfortunate cat whose mortality depends on the outcome of a quantum measurement. If a quantum subsystem was described by such a linear combination $a\ket{\text{alive}}+b\ket{\text{dead}}$, interaction with light, air, and dust from the environment will rapidly entangle such a state with the environment, removing any interference between these two states of the cat. 
\end{example}

\begin{example}\label{basis problem}Decoherence can help with the preferred basis problem. When describing a quantum subsystem entangled with the environment, we can not describe the subsystem by a pure quantum state, and it is customary to use density matrices. Using density matrices in the double-slit experiment of Example \ref{2 slit}, the unmeasured photon is described by the pure state 
\[\frac 12(\ket{\text{left}}+\ket{\text{right}})(\bra{\text{left}}+\bra{\text{right}})\]
whereas the measured photon is described by the mixed state
\[\frac 12\ket{\text{left}} \bra{\text{left}}+\frac 12\ket{\text{right}}\bra{\text{right}}\ .\]
It is tempting to interpret this mixed state as a statistical mixture of the pure states, but this density matrix can be written as a convex combination of pure states in a number of inequivalent ways. This issue is referred to as the preferred basis problem, because different choices of Hilbert space basis for the measuring apparatus  gives different interpretations of the photon's density matrix. Decoherence helps to resolve this preferred basis problem if we   assert that different measurement outcomes correspond to macroscopically different `pointer' states, special states that rapidly entangle with the environment and are resistant to decoherence by entanglement with the environment. This idea is referred to as  einselection in \cite{Zurekexistential}. See Section  \ref{density matrices} for further discussion on interpreting density matrices as statistical ensembles.
\end{example}

\begin{remark} All versions of decoherence require a choice of division of $\mathcal H$ into a tensor product of quantum subsystems, and this division  should be regarded as extra information required to describe the outcome of decoherence. To explain the emergence of classical phenomena,  any decoherence-based explanation must either specify the division of $\mathcal H$ or justify why the division does not matter. For example, a simplistic assertion that  a density matrix is a statistical ensemble of states fails this test.

 In this paper, we do not choose a decomposition of $\mathcal H$ into subsystems, but  do assume extra information beyond standard quantum mechanics so that we can specify  where projection-statements are localised. This extra assumption is standard in quantum field theory, and is therefore well justified. In our setup,  a partition of a Cauchy surface corresponds to a  division of $\mathcal H$ into subsystems, but we are careful to ensure that our definitions do not require a chosen partition. The paper \cite{multiverse} uses such geometric information to make a more canonical choice of `environment' subsystem using causal diamonds. 
\end{remark}

 To convey the main idea of Quantum Darwinism, we  consider a drastically simplified example: If $\psi$ starts off approximately in the form 
 \[\psi\simeq a\otimes x_{1}\otimes x_{2}\otimes\dotsb\otimes x_{n}\ \ \ \ \text{ in }\mathcal H=\mathcal H_{0}\otimes \mathcal H_{1}\otimes\dotsb \otimes \mathcal H_{n} \]
 then after\footnote{In other sections, we have been using the Heisenberg picture of quantum mechanics, in which any time evolution acts on operators instead of the state $\psi$, but for this simple example, we are slipping into the Schr\"odinger picture, where the state evolves. So, we are assuming we are on Minkowski space, and have chosen some splitting of spacetime to get a unitary time evolution $U(t)$ on our Hilbert space. In the Heisenberg picture, each projection-statement $P$ has related time-shifted projection-statement $P_t= U(t)\circ P\circ U(t)^{-1}$ so that $U_{P_t}$ is $U_P$ shifted in time by $t$, but in this section, we evolve the state $\psi$ using $U(t)$.    } certain types of quantum interactions of $\mathcal H_{i}$ with $\mathcal H_{0}$,  $\psi$ is approximately in the form 
 \[\psi\simeq\sum_{i}\lambda_i a_{i}\otimes x_{1,i}\otimes x_{2,i}\otimes\dotsb\otimes x_{n,i}\ \ \ \ \text{ in }\mathcal H=\mathcal H_{0}\otimes \mathcal H_{1}\otimes\dotsb \otimes \mathcal H_{n} \]
 where $\{a_{i}\}$  form an orthonormal basis for $\mathcal H_{0}$ and the unit vectors $x_{k,i}$ and $x_{k,j}$  tend to point in different directions for $i\neq j$. This process of copying of information from $\mathcal H_{0}$ into the `environment' represented by $\mathcal H_{1}\otimes\dotsb\otimes \mathcal H_{n}$ is referred to as Quantum Darwinism by Zurek in several papers including \cite{QD,QDclassical}. An example of a Hamiltonian generating such a time evolution is $M\otimes H_{1}+M\otimes H_{2}+\dotsc +M\otimes H_{n}$ where $M:\mathcal H_{0}\longrightarrow \mathcal H_{0}$ has eigenvectors $\{a_{i}\}$ with distinct eigenvalues and $H_{k}$ is some Hamiltonian on $\mathcal H_{k}$. More generally, we could replace $M\otimes H_k$ with any interaction Hamiltonian commuting with $M$. This is  a simplification of any physically realistic situation, see \cite{JoosZeh, QDBM, QDS} for more realistic examples. 
  
   Note that after this interaction, the result of the projection $P_1:=\ket{a_{1}}\bra{a_{1}}$ on $\psi$ can be approximated by the result of any of the projections $X_{k}:=\ket{x_{k,1}}\bra{x_{k,1}}$. For a better approximation, we may take the product of a number of these $X_{i}$. Under our simplistic assumptions, the error $\norm {X_{1}\dotsc X_{n}\psi-P_1\psi}$ can be expected to decay roughly exponentially in the number $n$ of similar systems $\mathcal H_{i}$. Similarly, the projection $P_I:=\sum_{j\in I}\ket {a_j}\bra{a_j}$ onto some range in the spectrum of $M$ can be approximated by the projection $X_k$ on to the span of $\{x_{k,j}\text{ for }j\in I\}$, with a better approximation achieved by taking the product of these projections for all $k$.  We are again left with  the heuristic that the error  decays roughly exponentially,  proportional to the number $n$ of similar systems $\mathcal H_i$ involved. 
  
    If time evolution is given by a Hamiltonian commuting with a projection $P$  on $\mathcal H_0$, an intuitive description of the error  $\norm {X_{1}\dotsc X_{n}\psi-P\psi}$ is that it is the square root of the probability that no interactions occur that distinguish the image of  $P$ from its orthogonal complement. One thing to notice is that we can not expect the size of this error will be proportional to $\norm{P\psi}$, and it is unreasonable to expect that the error vanishes. 
 
 \begin{remark}
 Despite reversible unitary time evolution, decoherence is generally only expected to happen as time increases, due to the requirement that $\psi$ start off at least approximately unentangled\footnote{Although \cite{QDinhazy} implies that at least in some cases, decoherence will still occur if there is some entanglement.} in the $\mathcal H_{i}$ factors. So, if we expect decoherence to produce classical projection-statements, we are expecting some measure of entanglement to increase with the arrow of time. Note also that classical projection-statements  tend to have images that are less entangled. 
\end{remark}  

\begin{example} In the standard interpretation of measurement in quantum field theory, \cite{qftMeasurement}, an observer located in a spacetime region $U$ must use observables localised in $U$ to probe the system.  In particular, for this observer to be able to distinguish projection-statements $P$ and $\neg P$ there must be an observable $O$ localised in $U$ with different statistics on $P\psi$ and $\neg P\psi$. The observables localised in spatially separated regions commute, so can be combined to provide an observable that better distinguishes $P$ and $\neg P$, with the most improvement expected when the statistics of our spatially separated observables are independent.

 For example suppose that we have a spatially separated collection of $n$  projection-statements $P'_i$ with $\prob(P_i'\vert P)>\prob (P_i'\vert \neg P)+\delta$ for some $\delta>0$. We can use the number of observations of $P'_i$ to create a projection-statement $P'$, describing all outcomes where the number of observations is  above $n\delta/2+\sum_i\prob(P_i\vert \neg P)$. Then under the assumption that the statistics of $P'_i$ are independent, a rough estimate for the error is
\[\norm{(P-P')\psi}\leq  e^{-n \delta^2/2}\]
So, in this situation we expect the visibility of $P$ in a spacetime region $U$ to increase roughly linearly with the size of the spacetime region. 
\end{example}

How can we expect the error $\norm{(P-P')\psi}$ to depend on the regions $U_P$ and $U_{P'}$ where $P$ and $P'$ are localised? If we  believe that the error decays exponentially with the number of similar interactions, then how does the expected number of these interactions depend on the geometry of $U_P$ and $U_{P'}$? In non-extreme situations, say when $P$ describes a not-particularly-dense arrangement of matter,  the number of interacting particles can be expected to grow roughly in line with the volume of $U_P$. In extreme cases, when $P$ describes something dense enough to be opaque,  we expect a bound roughly proportional to the surface area of $U_P$.

 We arrive at the following heuristic.
%

\begin{heuristic} The degree to which a projection statement $P$ can be visible as a projection statement $P'$ within a region $U_{P'}$ depends on the geometry of $U_{P'}$ and $U_P$ as well as physical properties of $P$ and $\psi$. In particular, we expect the error $\norm{(P-P')\psi}$ to be smaller when $U_{P'}$ fills up more of the causal future of $U_P$ on some Cauchy surface. Supposing that $U_{P'}$ takes up a fixed proportion of the causal future of $U_{P}$ on some Cauchy surface, we expect smaller error $\norm{(P-P')\psi}$ when $U_{P}$ is larger. 
\end{heuristic}

Let us give an independent argument for a geometric restriction on  the visibility of projection-statements. Given a $d$ dimensional subspace of a $N$--dimensional Hilbert space, the projection of a random unit vector to this subspace has expected square norm $d/N$. In simple models of decoherence, such as \cite[section 4A]{Zurekexistential}, the error in decoherence is also roughly proportional to $1/N$, where the dimension of the environment Hilbert space is $N$. Similarly, if $P$ is approximated by $P'$ within some subsystem of dimension $N$, the best we can expect for the error $\norm{(P-P')\psi}$ is that it is roughly of order $1/\sqrt N$. In a lattice model of a quantum field theory, the dimension of the Hilbert space associated to a spatial region is proportional to the exponential of the volume of that region. In more sophisticated models including gravity, the dimension of this Hilbert space is expected to be proportional to the exponential of the surface area of that region; \cite{Bekenstein,Bousso_2000}. Accordingly, we expect that the error $\norm{(P-P')\psi}$ to decrease at most exponentially in the surface area of the region where $P'$ is localised.

Taking the above heuristic argument seriously, we arrive at the following optimistically precise  bound. This bound restricts how visible a classical projection-statement can be in spatially separated balls on a spacelike surface. 
\begin{heuristic}\label{visibility bound} For some universal constant $k$,  if $U$ and $U'$ are spatially separated balls with surface area $S$, and $P$ is a projection statement localised in $U$  we expect the visibility of $P$ in $U'$ to satisfy 
\[\text{Vis}_{U'}P\leq \ln\norm{P\psi}-\ln\norm{(1-P)\psi}+kS\] 
\end{heuristic}

In particular, we don't expect totally classical projection-statements to exist, outside some fanciful limit sending the number of similar interactions to infinity, because totally classical projection-statements need to be recorded perfectly in many spatially separated, arbitrarily small balls.  More importantly, we don't expect low probability classical projection-statements to have records localised in small balls.

 \section{Observers}\label{observer}

   It is worth spelling out what an `observer' means in this paper. Observers, such as people, cats, or classical measuring devices, are physical things, subject to the usual physical laws. For us, any properties of a `classical observer' must be described using classical projection-statements.   For example,  a projection-statement $O$ describing the state of mind of a person might  describe one firing pattern of neurons for our observer.\footnote{Similar ideas have been suggested by several other people, including  Zeh in \cite{Zehinterpret} and \cite{Zehconscious}. In \cite{neuron}, Tegmark estimates decoherence rates for neurons in order to show that interference effects between different firing patterns of neurons are unlikely to play an important role in a physical description of our minds.}
   
To model the classical information available to an observer, we can choose  a classical $\mathcal H$--valued measure algebra $(\mathfrak B,\phi)$ containing the observer's projection-statement $O$. According to this model, our observer knows about statements $b\in\mathfrak B$ such that $b\wedge O=O$. The set  $\mathfrak T_O\subset \mathfrak B$ of such statements forms a filter.
   Associated to $O$ we also have the ideal $\mathfrak I_O\subset \mathfrak B$ comprised of statements $b\in \mathfrak B$ such that $b\wedge O=b$.  Then, $(\mathfrak I_O,\phi)$ is a classical $\mathcal H$--valued measure algebra representing the statements in $(\mathfrak B,\phi)$ that our observer regards as possiblities.  Accordingly, this observer would assign probabilities $\norm{\phi(b)}^2/\norm{O\psi}^2$ to statements $b$ in $\mathfrak I_O$. Our observer can still reason about statements $b\in\mathfrak B$ that are not in $\mathfrak I_O$, but from our observer's perspective, the statement $b\wedge O\in\mathfrak I_O$ is logically equivalent to $b$. 
   
    As $(\mathfrak B,\phi)$ is only taken to be an approximately correct model, it is not quite accurate that our observer knows statements in $\mathfrak T_O$. Instead, we should say that the observer is sure of statements close to $\mathfrak T_O$, that is, classical statements $t\in\mathfrak B$ such that $\phi(t\wedge O)\approx \phi(O)$, or equivalently, statements with $\prob(t\rvert O)\approx 1$. Heuristic \ref{algebra record replacement} implies that this notion is approximately independent of our choice of classical $\mathcal H$--valued measure algebra containing $\mathcal O$.

   Continuing the simplified example from section \ref{Decoherence}, an observer's initial state might be described by $O$, where  \[O\psi\approx a\otimes x_{1}\otimes \dotsb \otimes x_{n}\otimes o \in \mathcal H_{0}\otimes\mathcal  H_{1}\otimes\dotsb\mathcal H_{n}\otimes \mathcal H_{o}\ .\] Following interaction of $\mathcal H_{0}$ with $\mathcal H_{i}$, we might have $O\psi\approx\lambda_i a_{i}\otimes x_{1, i}\otimes \dotsc x_{n,i}\otimes o$, as in section \ref{Decoherence}. Subsequent interaction of $\mathcal H_{o}$ with $\mathcal H_{k}$ might produce \[\sum_{i}O\psi\approx \lambda_i a_{i}\otimes x_{1,i}\otimes \dotsb \otimes x_{n,i} \otimes o_{i} \] 
   with the $o_{i}$ now approximately orthogonal, and new classical projection-statements $O_{i}\in\mathfrak I_O$ describing our observer's new state, with \[O_{i}\psi\approx a_{i}\otimes x_{1,i}\otimes \dotsb \otimes x_{n,i} \otimes \lambda_i o_{i}\ .\] This new projection-statement $O_i$ knows about $O$, in the sense that $O\in \mathfrak T_{O_i}$,  and the observer $O$ regarded $O_i$ as a possibility in the sense that $O_i\in\mathfrak I_O$.  The interpretation of the projection-statement $O_{i}$ is that the observer, initially described by $O$, has observed the classical projection-statement $P_{i}:=\ket{a_{i}}\bra{a_{i}}$. With this said, the reader should remember that we are not considering $O$ or $O_{i}$ as `really' true or false, despite the philosophical baggage  attached to such a description of the observer.

   An interaction Hamiltonian approximately producing the behavior above is $\sum_{k}M_{k}\otimes H_{o}$, where $M_{k}$ is an observable measuring $\mathcal H_{k}$ in our observer's location with distinct expectation values for different $x_{k,i}$, and $H_{o}$ is a Hamiltonian acting on $\mathcal H_{o}$. Under the assumption that  the number of interactions $n$, is extremely large, such an interaction Hamiltonian is approximated by $\sum_{i} P'_{i}m_{i}H_{o}$, where $P'_{i}$ is a projection-statement approximating $P_{i}$, as in section \ref{Decoherence}, and $m_{i}$ is the  expectation of $\sum_k M_{k}$ on $x_{1,i}\otimes\dotsb\otimes x_{n,i}$.
   
Although the situation above is drastically simplified, one important element is that our observer's interaction with the environment is local, and of a restricted type. Physical observers such as ourselves have limited options for interacting with and observing the outside world --- one reason to believe that we observe and act on information described by classical projection-statements is that \emph{we can}.\footnote{We do not ask observers to verify that a classical projection-statement is classical --- we are asserting that observers can detect classical projection-statements by acting on correlations in environmental information. In some sense, this is what animals do --- search for, and act on correlated information with useful predictive power.}

A second reason to believe that we observe classical projection-statements is that \emph{it is useful}. Because classical projection-statements are visible in many separate regions of spacetime,  they have widespread consequences --- a point made in \cite{QDclassical, Zurekreview}. Moreover, Lemma \ref{replacement}, with Heuristics \ref{algebra record replacement} and \ref{classical measure algebra} imply that we can expect the `probabilities'  we associate to logical combinations of projection-statements to approximately obey usual probability theory. Accordingly,  classical projection-statements have reliable and predictable consequences.

    \section{Towards a theory of natural probability}\label{NP section}

\

Let us explore the consequences of the following assumption: \emph{Our perception of reality  is a classical projection-statement.} 

In Section \ref{precise}, we constructed a probabilistic model $(X,\{\mathfrak F_U\},\mu)$ for classical information using totally classical projection-statements, however, from Heuristic \ref{visibility bound}, we don't actually expect totally classical projection-statements to exist without taking an unphysical limit. In Section \ref{approximate}, we explained how to obtain approximations to this classical situation. Starting with a  small, spatially separated collection $I$ of records of classical projection-statements, we constructed  a classical $\mathcal H$--valued measure algebra $(\mathfrak B_I,\phi)$. This $\mathcal H$--valued measure algebra is our classical model, and allows us to assign probabilities to logical combinations of these projection-statements. In Heuristic \ref{algebra record replacement}, we explained  that we expect these probabilities to be approximately independent of the records chosen, so we we have approximately canonical assignments of `probabilities' to logical combinations of classical projection-statements, and these `probabilities' approximately obey ordinary probability theory.

%

There is more structure in $(\mathfrak B_I,\phi)$ than is present in ordinary probability theory. In ordinary probability theory, the only relevant information is the probability $\norm{\phi(b)}^2$ of a statement  $b\in\mathfrak B_I$, but in natural probability, we also care about the visibility of $b$ in different regions of spacetime; see definitions \ref{statement visibility} and \ref{classical statement def}.  In particular, some statements $b\in\mathfrak B_I$ are classical, and visible in many spatially separated regions of spacetime. We show, in Heuristic \ref{classical measure algebra}, that logical combinations of classical projection-statements tend also to be classical projection-statements, so long as their probability is high enough, with Heuristic \ref{caveat} providing the caveat that the degree of classicality decreases with more complicated logical combinations. 

%
%

In ordinary probability, one must accept the circular intuition that low probability events are unlikely. In our setup, this intuition is justified by our assumption that we can only perceive classical projection-statements:  it is significantly harder for low probability projection-statements to be classical, and, to detect a low probability classical projection-statement, we must sample more of space-time. The sharpest plausible formulation of this principle is Heuristic \ref{visibility bound}, which conjectures  a geometric bound on the visibility of a projection-statement $P$ within two spatially separated balls of surface area $S$. In particular, the visibility of $P$ is bounded by $\ln\norm{P\psi}-\ln\norm{\neg P\psi}+kS$ for some universal constant $k$.  So, the visibility of $P$ in one of the two balls must be less than this bound, and this bound is lower when $\prob(P):=\norm{P\psi}^2$ is smaller. We conclude the following.

\begin{heuristic}\label{classical cutoff} Each observer $O$ can only reliably know about classical projection-statements $P$ with $\prob(P)>p_O$, where $p_O$ depends on physical properties of the observer, including their size. Optimistically, there exists some universal constant $k$ so that an observer within a ball of surface area $S$ has $p_O>e^{-kS}$.
\end{heuristic}

There is no analogous statement within usual probability theory.

Observers of different size will have different cutoff probabilities. Let us imagine a microscopic computer called Tiny, and gigantic mega-computer called Bob. Tiny  will have a much larger cutoff probability, $p_{T}$, than Bob. Bob may know about a projection-statement $P$ with $\prob (P)<p_{T}$, so what is the problem with telling his friend, Tiny? An intriguing answer is that Tiny has limited memory, which can not fit the information required to reliably describe $P$. We arrive at the following surprising conclusion: 

\begin{heuristic}A classical projection-statement $P$ with $\prob(P)$ small requires a lot of information to reliably describe.
\end{heuristic} 

This is not so different from how we are meant to apply ordinary probability. To see a low probability event, all we need to do is flip a coin 1000 times, and record our answers. We do not, however, expect to see the easy-to-communicate event that all 1000 of those tosses were heads. If we partition our probability space into $4$ subsets, one of which has probability $2^{-1000}$, we do not expect that an experiment will end up with that result. If, however, we partition our probability space into $2^{1000}$ equal subsets, then an experiment distinguishing these outcomes is guaranteed to end up with an outcome that has probability $2^{-1000}$. Such a complicated experiment requires significantly more work, and communicating the outcome requires significantly more information. 

 The information storage capacity of an observer  depends on physical properties of the observer, including its size, with larger observers having more storage capacity. Measured in bits, an extreme estimate for the maximal information storage within a ball is the following: the information capacity is bounded by a universal constant times its surface area; see for example \cite{Bekenstein,Bousso_2000}. This  information capacity bound is similar to our proposed bound on visibility of projection statements, Heuristic \ref{visibility bound}.  So we could hope  that  there is some $\alpha>0$ such that, if  $P$ is a classical projection-statement taking $n$ bits to encode, then $\prob (P)>e^{-\alpha n}$. Put differently, a classical projection-statement $P$ requires at least $-\alpha\log(\prob (P))$ characters to reliably express in binary.

 To clarify, it is possible for Tiny and Bob to be using a system for encoding projection-statements so that some low-probability classical projection-statements  $P$ are easy to describe, but that would be a bad system, prone to errors. In particular, even if Bob observed $P$, and transmitted this information to Tiny, Tiny would not know $P$ in our sense:  there would not be a classical projection-statement $T$ describing Tiny such that $TP\psi\approx T\psi$, and predictions Tiny would make based on believing $P$ would not be reliable.  It is advantageous for animals to make reliable predictions, so we should have evolved a reliable system for encoding classical information.  For a binary version of such a system, we expect that $P$ takes at least $-\alpha\log(\prob (P))$ bits to reliably describe.  Using ordinary probability theory, Shannon, \cite{Shannon}, tells us that this is also an efficient system for storing and communicating information.

\

Let us summarise the differences between natural probability and ordinary probability theory. In ordinary probability theory, we have a probability space and some $\sigma$-algebra of events. We re-formulate this using  a measure algebra $(\mathfrak B,\mu)$. In natural probability,   we upgrade the measure algebra $(\mathfrak B,\mu)$ to a $\mathcal H$--valued measure algebra $(\mathfrak B,\phi)$, with $\mu(b)=\norm{\phi(b)}^2$.

In natural probability, we also have an extra notion of when a statement $b\in \mathfrak B$ is classical. Moreover, we expect small logical combinations of classical statements to also be classical, so long as they have sufficiently high probability. This departs from classical probability theory, where the only thing we know about an event is its probability.

 Within ordinary probability, we can specify a  subalgebra $\mathfrak B_U\subset \mathfrak B$ comprised of events accessible within a space-time region $U$. In natural probability,  we lose our sharp notion of such a  $\sigma$--algebra $\mathfrak B_U$.  We replace this with a more nuanced notion of the visibility $\text{Vis}_U b$ of a statement $b\in\mathfrak B$. It is not true that the statements with high visibility in $U$ form a $\sigma$-algebra, but small logical combinations of such statements  tend to also have high visibility in $U$, so long as they have sufficiently high probability.

 It is inconvenient to lose sharp, mathematically defined, $\sigma$--algebras $\mathfrak B_U$, but we also gain something. From Heuristic \ref{visibility bound}, we expect a geometric lower bound on the probability of classical statements visible within $U$ --- this is something entirely absent from ordinary probability theory. We can also make the assertion that an observer only perceives a classical statement, visible in their location. In ordinary probability theory, this is analogous to the deeply problematic, but intuitively appealing statement `you will never see a low probability event.'

\section{Interpretations of measurement and quantum mechanics}
\label{measurement}

Let us describe a measurement process from our perspective. We need a classical projection-statement $E$ describing the experimental setup, including the initial state $s$ of an isolated quantum system $\mathcal H_{s}$ to be measured, and the initial state of a measuring apparatus $a$, (so $E\psi\approx s\otimes a\otimes\dotsb $). This is often described as preparing the system $\mathcal H_{s}$. The measuring apparatus then interacts with our previously isolated quantum system in a unitary process often called pre-measurement or a non-selective measurement. After this interaction, 
\[E\psi\simeq\sum_{i}\lambda_i s_{i}\otimes a_{i}\otimes\dotsb\] 
where $a_{i}$ are orthonormal, and $s_{i}$ form an orthonormal basis for $\mathcal H_{s}$. The Copenhagen interpretation of quantum mechanics asserts that at some time  $E\psi$  `collapses' to  one of these  summands $s_{i}\otimes a_{i}\otimes\dotsb$ with probability $\abs{\lambda_i}^2 /\norm{E\psi}^{2}$.

The states $a_i$ of the measuring apparatus are special. To be distinguished by classical observers, these states of the measuring apparatus $a_{i}$ need to be `pointer states' distinguishable through interaction with the environment,  and for  sufficiently likely measurement outcomes,  there are classical projection-statements $M_{i}$ that describe the state of our measuring apparatus so that $M_{i}E\psi\approx s_{i}\otimes a_{i}\otimes \dotsb$. The conditional probability $\prob(M_{i}\vert E)$ we associate to a measurement outcome agrees with the usual probability, however we do not have a notion of only one outcome being really true, and do not regard $M_{i}$ as describing a physical process.  

\

More formally, a measurement is often described  in terms of a projection-valued measure, allowing observables with continuous spectrum.  We offer the following  model for such measurements.

\

\begin{defn}\label{measurement model} A model of a quantum measurement is 
\begin{itemize}
\item a  classical $\mathcal H$--valued measure algebra $(\mathfrak M, \phi)$, with maximal element $E$ describing the experimental setup,
\item a decomposition of $\mathcal H$ into $\mathcal H_s\otimes \mathcal H_e$, and 
\item a projection-valued measure $\pi$ on $\mathfrak M$, such that, for $M\in \mathfrak M$, $\pi(M)$ splits as a projection on $\mathcal H_s$ tensored with the identity on $\mathcal H_e$ and so that 
\[\pi(M)\phi(E)\approx \phi(M) \]
\end{itemize}
 \end{defn}

 The relative probability $\norm{\phi(M)}^2/\norm{\phi(E)}^2$ of the measurement outcome $M$ is the same as the customary probability  arising from the projection-valued measure $\pi$ acting on  $\mathcal H_s$. The above model adds extra information beyond a projection--valued measure $\pi$ on  a set of outcomes, $X$. In particular, we don't expect \emph{all} measurable subsets of $X$ to correspond to valid outcomes of this measurement, as not all of these will be represented by classical projection-statements. Instead, from Definition \ref{new classical measure algebra def}, we expect that each outcome $M\in\mathfrak M$ is close to a classical outcome $M'$, using the metric on $\mathfrak M$ induced from $\phi:\mathfrak M\longrightarrow \mathcal H$.
%
%
%
%
%

\

  From our perspective, several interpretations of quantum measurement have partial validity. We mention some such interpretations below:

  \subsection{``A measurement is an update of an observer's information."} Such an interpretation is appealing, because in ordinary probability, updating information causes a collapse in the probability distribution describing an observer's knowledge. The obvious problem with such an interpretation is that, unlike the classical case,  we can't treat quantum measurements as dynamically irrelevant, because quantum measurements don't commute --- for example a rotating sequence of quantum measurements can make a given measurement outcome imply its opposite with arbitrarily high probability. However, with more careful treatment we see the validity of such an interpretation. For us, an observer, described by a classical projection-statement $O$, might know the experimental setup $E$. If the observer then notices the outcome $M$ of our experiment, our observer's projection statement is updated to  $O'$  that knows, $O$, $E$, and $M$. One could say that our observer's knowledge was updated to include the result $M$ of the measurement.

For all practical purposes, we can regard the relative state $O\psi$ used by an observer as an epistemological state, describing the knowledge available to the observer $O$. When someone asserts that `the quantum state must be information about the potential consequences of our interventions in the world', \cite{Fuchs}, they are referring to this relative state $O\psi$.  In contrast, we can regard $\psi$ is a valid ontological state, describing Nature in entirety.  Despite these fancy words, neither of these states quite fit into philosophical frameworks developed by ancient Greeks. The identity projection-statement, corresponding to the ontological state $\psi$, should be regarded as a tautology, saying nothing. In contrast,  when a projection-statement $P$ is more specific, $P\psi$ is smaller, and can be regarded as describing less of the ontological state $\psi$. We can think of $O\psi$ as analogous to a perspective of $\psi$, with more detailed perspectives containing less of the whole. 
 In classical probability, this is analogous to thinking of the probability measure on the entire space as the ontological state, which is dual to the customary habit of thinking of some `really true' point in the probability space as being the ontological state, and thinking of the probability measure as an epistemological state. 

  
  \
  
   \subsection{``The measurement is an interaction of the isolated quantum system with classical reality.''}
  The obvious problem with such an interpretation is that it begs the question `what is classical reality?', and requires us to somehow choose a divide between the quantum and classical worlds. We interpret `classical reality' as the stuff described by classical projection statements. Then, our description of measurement above can indeed be interpreted an interaction of our isolated quantum system $\mathcal H_{s}$ with the measuring apparatus and the aspects of Nature described by classical projection-statements.

 \subsection{``Measurement collapse has something to do with the observer's consciousness.'' }  
  We do not think wave-function collapse is a physical process, however describing an observer's conscious state as a classical projection-statement $O_{i}$ allows us to say that the observer `knows that $E\psi$ has collapsed into $M_{i}E\psi$'. Thus, the observer's perception of wave-function collapse is produced by our choice of how to describe the observer's conscious state. We could also describe all probabilistic aspects of our quantum theory purely in terms of observer's states --- this is commonly known as the `many minds interpretation' \cite{MMI}, and is  called `sensible quantum mechanics' in \cite{sqm}.

\subsection{``Density matrices represent statistical ensembles of states, which can be  updated with the new information obtained by the measurement."}  \label{density matrices}

 Given a state $\psi\in \mathcal H_s\otimes \mathcal H_e$ we can forget information from the environment $\mathcal H_e$ by performing a partial trace of  $\ket\psi\bra\psi$ to obtain semi-positive operator $\rho_\psi$ on $\mathcal H_s$ called a density matrix. Because $\rho_\psi$ can be written as a convex combination of projections onto pure states, it is common to interpret $\rho_\psi$ as a statistical ensemble of states. This interpretation is ambiguous because $\rho_\psi$ can be written as a convex combination of pure states in many ways, corresponding to many different interpretations of $\rho_\psi$ as a statistical ensemble of states. 
 
  A  measurement of $\mathcal H_s$ is sometimes described as follows: \emph{ First, $\mathcal H_s$ interacts with $\mathcal H_e$ in a nonselective measurement resulting in a state in the form $\sum_i \lambda _i s_i\otimes y_i$, where $y_i$ are orthonormal. Following this $\rho_\psi=\sum_i \abs{\lambda_i}^2 \ket {s_i}\bra {s_i}$ is the  statistical ensemble representing the possible outcomes of the measurement, and we are free to update this statistical ensemble once we know the outcome of the measurement.} 
  
  Here, we emphasise again  that the interpretation of $\rho_\psi$ as a convex combination of pure states depends on the choice of orthonormal bases $y_i$ of $\mathcal H_e$, so the interpretation of a density matrix as an ensemble of states is ambiguous. As outlined in Example \ref{basis problem}, decoherence helps resolve this ambiguity. In what follows, we consider how accurate it is to think of a density matrix as a statistical ensemble of states.
  
  Throughout this section, we measure the distance between density matrices using the Hilbert space inner product on $\mathcal H\otimes\mathcal H^*$.  On trace-class operators such as density matrices, this inner product coincides with the trace inner product
  
  \[ \langle f,g\rangle:=Tr (f^*g)\]
  
We model classical information as a $(\mathcal H,\psi)$--valued measure algebra  $(\mathfrak B,\phi)$. To a non-zero statement $b\in\mathfrak B$, we can associate the density matrix
\[\rho(b):= \frac 1{\norm{\phi(b)}^2}\rho_{\phi(b)}\]
where $\rho_{\phi(b)}$ indicates the partial trace of $\ket{\phi(b)}\bra{\phi(b)}$.
 We need further assumptions to think of $\rho_\psi$ as representing a statistical ensemble.
 
 \begin{defn} A partition of  a statement $b\in\mathfrak B$, is a collection $\{b_i\}\subset\mathfrak B$ such that $\bigvee_i b_i =b$ and $b_i\wedge b_j=0$ for $i\neq j$. Say that a partition $b=\bigwedge_i b_i$ is separate from $\mathcal H_s$ if there exist commuting projections $1\otimes P_i$ on $\mathcal H_s\otimes\mathcal H_e$ such that $\phi(b_i)=(1\otimes P_i)\phi(b)$. Say that $\mathfrak B$ is separate from $\mathcal H_s$ if every partition of $1\in\mathfrak B$ is separate from $\mathcal H_s$.
\end{defn}

For example, if $(\mathfrak B,\phi)$ is a classical measure algebra, constructed as in section \ref{approximate} using spatially separated projection-statements on some spacelike surface, we could imagine that $\mathcal H_s$ represents a quantum subsystem located somewhere on this surface, with $b$ representing some description of an experimental setup isolating  $\mathcal H_s$, and the projection statements $1\otimes P_i$ spatially separated from $\mathcal H_s$. 

 How does being separate from $\mathcal H_s$ help? If $\phi(b)= 1\otimes P(\psi)$, then $\rho_\psi=\prob(b)\rho(b)+\prob(\neg b)\rho(\neg b)$. Accordingly,  under the assumption that the partition $b=\bigwedge b_i$ is separate from $\mathcal H_s$,  we can interpret $\rho(b)$ as the following statistical ensemble of density matrices. 
\begin{equation}\label{additive density}\rho(b)=\sum_i \prob(b_i)\rho(b_i)\end{equation}

Let us clarify what is meant by a statistical ensemble of density matrices or states on $\mathcal H_s$, and what we mean by statistical ensembles being close

\begin{defn} \label{ensemble def}A statistical ensemble of density matrices is a measurable map
\[\rho:X\longrightarrow (\text{density matrices on }\mathcal H_s)\]
where $(X,\mu)$ is a probability space. If the image of $\rho$ consists of pure states, we say that this is a statistical ensemble of pure states. 
\end{defn}
\begin{defn}
A fine-graining $(X',\mu',\rho')$ of a statistical ensemble of density matrices $(X,\mu,\rho)$  is a measure preserving map 
\[\pi:(X',\mu')\longrightarrow (X,\mu)\]
such that, for all measurable subsets $S\subset X$ with $\mu(S)>0$, the expectation of $\rho$ on $S$ is the same as the expectation of $\rho'$ on $\pi^{-1}(S)$. 
\[\frac 1{\mu(S)}\int_S\rho d\mu=\frac 1{\mu'(\pi^{-1}(S))}\int_{\pi^{-1}S}\rho'd\mu'\]
\end{defn}
\begin{defn}
Given a fine graining, we measure the distance between the statistical ensembles $\rho$ and $\rho'$ using the $L^2$ norm on $(X',\mu')$.
\[\norm{\rho-\rho'}_2:=\sqrt{\int_{X'}\norm{\rho(\pi(x))-\rho'(x)}^2d\mu'(x)}\]
If this distance is $0$, then we say that the fine-graining is an equivalence.

More generally, an approximate equivalence between two statistical ensembles $(X_i,\mu_i,\rho_i)$ is  a diagram
\[\begin{tikzcd} (X_1,\mu_1)\rar{\rho_1} & (\text{density matrices on }\mathcal H_s)
\\ (X,\mu)\uar{\pi_1} \rar{\pi_2} & (X_2,\mu_2)\uar{\rho_2}\end{tikzcd}\]
such that  $\pi_i :(X,\mu)\longrightarrow (X_i,\mu_i)$ is an approximately measure preserving map, and 
\[\int_{X}\norm{ \rho_1(\pi_1(x))-\rho_2(\pi_2(x)))}^2 d\mu\approx 0\] 
\end{defn}

\begin{lemma} \label{density ensemble} When $(\mathfrak B,\phi)$ is separate from $\mathcal H_s$, there is a canonical statistical ensemble of density matrices $(X,\mu,\rho)$ such that $(\mathfrak B,\frac{\norm{\phi}^2}{\norm{\phi(1)}^2})$ is the measure algebra of $(X,\mu)$, and for all $b\in\mathfrak B$, $\rho(b)$ is the expectation of $\rho$ over the set $S_b\subset X$ corresponding to $b$. 

The assignment of a statistical ensemble of density matrices is functorial in the following sense.  A fine graining $(\mathfrak B,\phi)\hookrightarrow (\mathfrak B',\phi)$ is associated with fine graining $(X',\mu',\rho')\longrightarrow (X,\mu,\rho)$, and  an ideal $\mathfrak I_b\subset\mathfrak B$, is associated with the  Bayesian update $(S_b,\frac{\mu}{\mu(b)},\rho)$ of the probability measure on the ensemble.
\end{lemma}
\begin{proof}
 Equation \ref{additive density} implies that the assignment $b\mapsto \prob(b)\rho(b)$ determines a countably additive measure on $\mathfrak B$, valued in semi-positive operators on $\mathcal H_s$. Moreover,  $\norm{\prob(b)\rho(b)}=\mu(b)\norm{\rho(b)}\leq \mu(b)$ so this measure also has bounded variation. So,  using the  Radon-Nikodym theorem we get a random variable, or measurable map  
\[\rho:X\longrightarrow (\text{semipositive operators on }\mathcal H)\]
such that $(X,\mu)$ is the probability space with measure algebra $(\mathfrak B, \mu)$ with $\mu(b)=\norm{\phi(b)}^1/\norm{\phi(1)}^2$, and $\rho(b)$ is the expectation of $\rho$ on the set $S_b\subset X$ corresponding to $b\in\mathfrak B$. Here, $X=\hom( \mathfrak B,\{0,1\})$, $S_b$ is the set of homomorphisms sending $b$ to $1$,  and the measure $\mu$ on $X$ is constructed in Lemma \ref{space from algebra}. As $\rho(b)$ has trace $1$ for all $b$, the same holds almost everywhere for $\rho$, and $\rho$ can be taken as a measurable map to the space of density matrices.

Let us show that this construction is functorial. As $X=\hom(\mathfrak B,\{0,1\})$, the assignment of $X$ to $\mathfrak B$ is a contravariant functor, and as $\mu(b)=\norm{\rho(b)}/\norm{\rho(1)}$, if $(\mathfrak B,\phi)\hookrightarrow(\mathfrak B',\phi)$ is a boolean algebra homomorphism preserving $\phi$, the corresponding map $(X',\mu)\longrightarrow (X,\mu)$ is a fine-graining. Moreover, as $\rho(b)$ is the expectation of $\rho$ on the set $S_b$ of homomorphisms sending $b$ to $1$, we have that this map is a fine-graining of statistical ensembles of density matrices, in the sense of Definition \ref{ensemble def}. So, this construction provides a functor from the category of $\mathcal H$--valued measure algebras separate to $H_s$, to the category of statistical ensembles of density matrices. 

We must also show that an ideal $\mathfrak I_b\subset \mathfrak B$ corresponds to a Bayesian update. We have that $S_b$ is in one-to-one correspondence with $\hom(\mathfrak I_b,\{0,1\})$ as each such homomorphism must send $b$ to $1$, and extends uniquely to a homomorphism from $\mathfrak B$ sending an arbitrary element $b'$ to the same thing as $b'\wedge b\in\mathfrak I_b$. Moreover, the ensemble $(S_b,\mu',\rho')$ associated with $\mathfrak I_b$ has $\mu'=\mu/\mu(b)$, and $\rho'=\rho$, agreeing with the usual Bayesian rule for updating a probability measure upon observing an event $b$.

\end{proof}

What is special about pure states compared to other density matrices? Pure states can not be written as a convex combination of other density matrices, so,  any fine-graining of a statistical ensemble of pure states is actually an equivalence. 
When the average von Neumann entropy of an ensemble of density matrices is low, the ensemble is approximated by an ensemble of pure states. Moreover, in this case, each fine graining is also an approximate equivalence, so we will argue that in this case, it is approximately true that $\rho_\psi$ is canonically a statistical ensemble of density matrices.

We can use either the von Neumann entropy $S(\rho):=-Tr(\rho\ln\rho)$ or  the linear entropy 
\[S_L(\rho)=1-Tr\rho^2=\langle \rho,(1-\rho)\rangle\] as a convenient measure  of how close a density matrix is to being pure. We have $S(\rho)\geq S_L(\rho)$ vanishes only on pure states, and is positive on other density matrices. When we write
\[\rho=\sum_i \lambda_i \ket{a_i}\bra{a_i}\]
with $a_i$ orthonormal, we have
\[S(\rho)= -\sum_i \lambda_i\ln \lambda_i \]
and 
\[S_L(\rho)=\sum_i\lambda_i(1-\lambda_i)= 1-\sum_i\lambda_i^2\ .\]
Supposing that  $\lambda_1\geq\lambda_i$ for $i>0$, we can use this to estimate how close $\rho$ must be to the pure state $\ket{a_1}\bra{a_1}$. 
\begin{equation}\label{lambda estimate} S(\rho)\geq S_L(\rho)\geq \lambda_1(1-\lambda_1)+\sum_{i>1}\lambda_i(1-\lambda_1)=1-\lambda_1\end{equation}
so, 
\begin{equation}\label{entropy estimate}\begin{split}\norm{\rho-\ket{a_1}\bra{a_1}}^2&=(1-\lambda_1)^2+\sum_{i>1}\lambda_i^2
\\ &=2(1-\lambda_1)-S_L(\rho)
\\ &\leq S_L(\rho)\end{split}\end{equation}
 and we conclude that $\rho$ is close to the pure state $\rho_{\text{pure}}:=\ket{a_1}\bra{a_1}$ when the entropy $S(\rho)$ is small.
\begin{equation}\label{pure estimate}\norm{\rho-\rho_{\text{pure}}}^2\leq S_L(\rho)\end{equation}

\begin{lemma}\label{partition density matrix} Suppose that we can write the density matrix $\rho$ as a convex combination of density matrices.  $\rho=\sum_i a_i\rho_i$. Then,
\[\sum a_i\norm{\rho-\rho_i}^2=S_L(\rho)-\sum_ia_iS_L(\rho_i)\leq S_L(\rho)\ ,\]
so the average square distance betwee $\rho$ and $\rho_i$ is the drop in average linear entropy. 
\end{lemma}
\begin{proof} 
Taking the inner product of $\rho$ with the equation $\rho=\sum_ia_i\rho_i$, we get the following.
\[\begin{split}\langle\rho,\rho\rangle&=\sum_ia_i\langle\rho, \rho_i\rangle 
\\ &=\sum_i\frac {a_i}2(\norm\rho^2+\norm{\rho_i}^2-\norm{\rho-\rho_i}^2)
\end{split}\]
Using that $\sum_i a_i=1$,   we get 
\[\sum_ia_i\norm{\rho-\rho_i}^2=-\norm\rho^2 + \sum_ia_i\norm{\rho_i}^2\]
then, noting that $S_L(\rho)=1-\norm\rho^2$, we have the required expression.
\[\sum a_i\norm{\rho-\rho_i}^2=S_L(\rho)-\sum_ia_iS_L(\rho_i)\leq S_L(\rho) \]

\end{proof}

\begin{defn}   The average linear entropy of a partition $b=\bigvee_i b_i$ is
\[S_L\left(b=\bigvee_i b_i\right):= \sum_{i}\prob(b_i \vert b)S_L(\rho(b_i))\ .\]
The average linear entropy of an ideal $\mathfrak I_b\subset \mathfrak B$ is 
\[S_L(\mathfrak I_b):=\inf S_L\left (b=\bigvee_i b_i\right)\]
where the infimum is over all partitions of  $b$ that are separate from $\mathcal H_s$.

Similarly, given a statistical ensemble of density matrices $(X,\mu,\rho)$, the average linear entropy of this statistical ensemble is
\[S_L(X,\mu,\rho):=\int_{X}S_L(\rho(x)) d\mu(x)\]
\end{defn}

If $(\mathfrak B,\phi)$ is separate from $\mathcal H_s$, we have that \[S_L(\mathfrak B)=S_L(X,\mu,\rho)\ ,\] where $(X,\mu,\rho)$ is the statistical ensemble of density matrices constructed from $(\mathfrak B,\phi)$ in Lemma \ref{density ensemble}.

As a consequence of Lemma \ref{partition density matrix} we have the following. 
\begin{cor}\label{fine graining estimate} If $(\mathfrak B',\phi)$ is a fine-graining of $(\mathfrak B,\phi)$, separate from $\mathcal H_s$, then the distance between the corresponding statistical ensembles of states $(X,\mu,\rho)$ and $(X',\mu,\rho')$ is determined by the drop in average linear entropy.
\[\norm{\rho-\rho'}^2_2=S_L(\mathfrak B)-S_L(\mathfrak B')\]

\end{cor}

So, any fine graining of a statistical ensemble of density matrices is an approximate equivalence, so long as the average linear entropy is low. 

%
%
\begin{lemma}\label{pure approximation} Suppose that $(\mathfrak B,\phi)$ is separate from $\mathcal H_s$. Then, for any $\epsilon>0$  the statistical ensemble of density matrices $(X,\mu,\rho)$ from Lemma \ref{density ensemble} is approximated by a statistical ensemble of pure states $(X,\mu,\rho_{\text{pure}})$ with
\[\norm{\rho-\rho_{\text{pure}}}_2\leq \sqrt{S_L(\mathfrak B)}+\epsilon\]
\end{lemma}
\begin{proof} Choose a partition $1=\bigvee_i b_i$ with $\sqrt{S_L\left(1=\bigvee_i b_i\right)}\leq \sqrt{S_L(\mathfrak B)}+\epsilon/2$.  There is a corresponding partition $X=\coprod_i S_{b_i}$. As $\rho$ is integrable and bounded, we can also choose our partition fine enough  that there exists some simple function $\rho_s$ which is constant on each $S_{b_i}$ such that
\[\norm{\rho_s-\rho}_2\leq  \epsilon/2\]
Moreover, we can take $\rho_s$ equal to $\rho(b_i)$ on $S_{b_i}$. Define the statistical ensemble $\rho_{pure}$ to be the pure state $\ket{b_i}\bra{b_i}$ satisfying inequality (\ref{pure estimate}) on $S_{b_i}$. We then have 

\[\begin{split}\norm{\rho-\rho_{\text{pure}}}_2\leq &\norm{\rho_s-\rho_{\text{pure}}}_2+ \epsilon/2
\\ & = \epsilon/2+\sqrt{\sum_i \mu(b_i)\norm{\rho(b_i)-\ket{b_i}\bra{b_i}}^2}
\\ & \leq \epsilon/2+\sqrt{\sum_i \mu(b_i) S_L(\rho(b_i))}
\\ & \leq \sqrt{S_L(\mathfrak B)}+\epsilon \end{split}\]

\end{proof}

With the above discussion in mind, when does a density matrix represent a statistical ensemble of pure states? Suppose that $\mathcal H_s$ describes an isolated quantum subsystem in some spacetime region $U$. We model the classical information available outside of $U$ as a classical $(\mathcal H,\psi)$--valued measure algebra $(\mathfrak B,\phi)$, generated by projection statements spatially separated from  $U$. So, $(\mathfrak B,\phi)$ is separated from $\mathcal H_s$, and using Lemma \ref{density ensemble} we get a statistical ensemble $(X,\mu,\rho)$ of density matrices. Moreover, this ensemble is  close to a statistical ensemble of pure states if the average linear entropy $S_L(\mathfrak B)$ is low; Lemma \ref{pure approximation}.  This statistical ensemble is not canonical, however using Heuristic \ref{algebra record replacement}, we expect any two choices of $(\mathfrak B,\phi)$ to have an approximate common fine-graining, leading to an approximate common fine graining of statistical ensembles of density matrices. Moreover, if the average linear entropy is low, Corollary \ref{fine graining estimate} implies that any fine-graining is an approximate equivalence. So, assuming that $S_L(\mathfrak B)$ is sufficiently small, we have an approximately canonical statistical ensemble of pure states, representing the classical knowledge of $\mathcal H_s$. Similarly, given an observer $O\in\mathfrak B$, if the ideal $\mathfrak I_O$ has $S_L(\mathfrak I_O)$ small, then the observer's knowledge of $\mathcal H_s$ is accurately modelled by a statistical ensemble of pure states. 

It is not necessarily true that there will exist a classical $(\mathcal H,\phi)$--valued measure algebra $(\mathfrak B,\phi)$ separate from $\mathcal H$ with low average linear entropy, so it is not necessarily true that there is a statistical ensemble of pure states approximating the classical knowledge of $\mathcal H_s$. Given a model $(\mathfrak M,\phi)$ for a quantum measurement with discrete outcomes $M_i\in\mathfrak M$ as described above Definition \ref{measurement model}, we expect the linear entropy of $\rho(M_i)$ to be approximately $0$. So, in this case, $S_L(\mathfrak M)\approx 0$. If $E\in \mathfrak M$ is the maximal element describing the experimental setup, it is then accurate to think of $\rho(E)$ as a statistical ensemble in the usual way.

\subsection{``Reality branches into separate possibilities upon measurement, as in some many-worlds interpretations.'' }
  At first glance, such an interpretation is appealing because it appears to circumvent wave-function collapse as a separate physical process. However, if we must specify exactly when and how reality branches, we are left describing a separate physical process, which we do not want to do. Accordingly, modern many-worlds interpretations are more nuanced, \cite{MWI,WallaceGeneral,WallaceBook,Vaidman},  as was Everett's original relative state formulation, \cite{Everett}. From our perspective, any belief in separate branches is an idealization of our fuzzier setting, where classical projection-statements need not be distinct, and any `branching' is local,  gradual, and not sharply-defined. An important point is that each classical projection-statement is local (as in Everett's original version \cite{Everett}), and so must correspond to a superposition of  branches that differ elsewhere in spacetime. Our natural probability is only useful if we make such a superposition of branches.

   One sophisticated  many-worlds interpretation is advocated by Griffiths, Omn\`es,  and Gell-Mann and Hartle using consistent or decohered histories \cite{GMH, StrongDecoherence, Griffiths, Omnes}. The general idea is to choose a sequence of times, and at each time, a set of mutually orthogonal projection operators summing to the identity --- the different projection operators correspond to different things that can happen at that time. A history is then a choice one of these projections at each time, or more generally, a sum of such choices. There are various definitions of when histories are consistent or decohere, so that probabilities may be consistently assigned to them, and given a decoherent set of histories,  the temptation is to say that we only experience one of these histories. A key difficulty with this interpretation is that it depends on the choice of projection operators, so we are left to search for canonical choices that will include the kind of classical history we experience.\footnote{Alternately, one could apply unfamiliar logical foundations as in \cite{Isham}.} For example, in \cite{StrongDecoherence}, Gell-Mann and Hartle reduce their choices by demanding a strong version of decoherence, and striving for a maximally fine `quasiclassical' set of decohering histories. One problem with finding canonical choices of projection operators is that any realistic mechanism of decoherence is \emph{approximate}, however authors generally require an \emph{exact} decoherence condition to apply usual probability theory. This problem can be circumvented by tweaking projections so that exact decoherence holds, however  such tweaking will thwart our aspirations for canonical projections. 
   
   How does our approach compare to decohered histories? We demand a canonical set of projection operators\footnote{So if we were to define a set of histories, we would restrict ourselves to choosing projections that are projection-statements.} --- and call them projection-statements. For this, we pay the price that our decoherence is necessarily \emph{approximate} instead of exact. This price means that we can not use usual probability theory, but   our theory of natural probability relies on decoherence not being exact.
 The histories we are interested in are small logical combinations of classical projection-statements, corresponding to extremely coarse histories. We make no attempt to find a canonical, maximally fine set of (approximately) decohered histories --- our natural probability is not based on the idea of choosing from a configuration space of possibilities, and from our perspective, searching for a  canonical, maximally fine space of alternate histories is searching for a possibly non-existent idealization. 
    
    \subsection{Relationship with quantum Darwinism}
    
Let us sketch the relationship between natural probability and Zurek's Quantum Darwinism, \cite{QD}. We follow Zurek when we demand that  information about a classical projection-statement is redundantly stored in several different places. Unlike Zurek, we do not use a fixed division of the quantum Hilbert space into systems, instead using projection-statements localised in particular regions in spacetime. This is similar, but we are interested how information from a spacetime region $U$ is recorded in regions within its causal future --- even if we assign a quantum subsystem to each spacetime region, these systems should only be regarded as separate if they are spatially separated. Zurek uses quantum mutual information as a convenient measure of how much information a fragment of the environment knows about a separate subsystem. We instead use the visibility of projection statements as a measure of how well they can be accessed within a given region.  

The major difference between natural probability and Zurek's Quantum Darwinism is our treatment of probability. In section IV of \cite{QD}, Zurek argues using symmetry that the usual probabilities may be assigned to measurement outcomes. He basis this on the idea that equal probabilities should be assigned to outcomes related by a symmetry. For us, Zurek's  symmetry argument fails,\footnote{There are similar symmetry arguments for obtaining the usual Born probability rule, such as an argument using decision theory \cite{Deutsch,WallaceBorn}, or the use of self locating uncertainty in  \cite{Sebens_2018}. Like in Gleason's theorem, \cite{Gle},  such arguments generally find that the quantum situation is better than classical probability, because the quantum situation is more constrained. } because equally probable projection-statements might still have differences in how well they are recorded, so there is no symmetry exchanging them that preserves all relevant structure.  This important difference is lost in an ideal limit such as in Section \ref{precise}. 

    
   \section{Further questions}
   
   \subsubsection{Including gravity is difficult}
   \label{gravity}
To define classical projection-statements, we needed a  fixed spacetime so that we could say when a projection-statement was localised in a given spacetime region. We relied heavily on the assumption that spatially separated projection-statements commute, so needed to at least use the causal structure of spacetime. The problem with this is that Einstein's equations imply that the causal structure of spacetime depends on the distribution of mass. One could imagine moving a large mass around, depending on the outcome of a quantum measurement, thus affecting gravity and causal structure of spacetime. As such it is reasonable to expect that the causal structure of spacetime itself is an emergent classical approximation.

As such, we need to modify the definition of a classical projection-statement to make sense in a quantum theory including gravity. One could imagine a quantum theory of gravity on a fixed spacetime, where the spacetime metric is an observable. In such a theory, it would still make sense to say that a projection-statement is localised in a region of spacetime, but we must be more careful when asserting that two projection-statements are spatially separated. In particular, we could assume that there are gravitational projection-statements $G$ specifying properties of the metric in a spacetime region. We could also imagine that our Hilbert space splits as $\mathcal H\otimes\mathcal G$, so we have gravitational projection statements from $\mathcal G$ commuting with other types of projection-statements from $\mathcal H$.  Then if $G$ is a gravitational projection-statement roughly specifying the metric in a spacetime region $U$, we can talk about a collection of projection-statements $P_i$ within $U$ that are spatially separated using the metric information from $G$, so these projections commute when applied to $G\psi$. We can then repeat the arguments given earlier in this paper, with $G\psi$ playing the role of $\psi$. 

This gives a reasonable theory of classical projection-statements relative to $G$, but we should also include a suitable notion of $G$ being classical. One possibility is to require that $G$ be approximately recoverable from the structure of the classical projection-statements relative to $G$, just as the spacetime metric can be approximately recovered by indirect measurements. Whenever $P_i G\psi\approx P_iP_jG\psi$ there is a corresponding linear subspace of $\mathcal H$ on which $P_i=P_iP_j$. We could demand that there is some such collection of implications between classical projection-statements relative to $G$ and a corresponding linear subspace of  $\mathcal H$, such that orthogonal projection $\pi$ to this subspace satisfies $\pi\psi\approx G\psi$.

So, it is conceivable that there is a version of natural probability compatible with quantum gravity. With that said, the difficulties of quantum gravity are such that we leave this as a formidable exercise for the ambitious reader.

\subsubsection{Realistic derivation of classical projection-statements}

In Section \ref{Decoherence}, we used a very simple model to argue that, in some quantum systems,  classical projection-statements are created by decoherence-like processes. The obvious question is: \emph{Does this also happen in more realistic models of Nature?} There are many realistic treatments of decoherence,  \cite{survey,JoosZeh}, and Quantum Darwinism \cite{QDBM,QDinhazy,QDS,QDclassical}, but in this paper, we have not demonstrated that classical projection-statements are produced in the most accurate and generally accepted quantum models of Nature. 

Given a realistic quantum model of Nature,  it would  be interesting to rigorously demonstrate the emergence of almost deterministic classical mechanics. Some aspects are very well established --- for example, we can use the Wigner-Weyl transform to connect functions on phase space with wave functions, and the time evolution in this phase space formulation of quantum mechanics involves the Moyal bracket, which is a deformation of the usual Poisson bracket determining classical time evolution on phase space; \cite{Moyal_1949}. So, gaussian wavepackets will evolve approximately along classical trajectories in phase space. This in itself does not explain the emergence of classical mechanics; see for example \cite[Section III]{Zurekreview}. There has also been considerable work showing that separated wavepackets decohere; see, for example \cite[Section V]{Zurekreview}. For the emergence of classical mechanics, we would like to see classical projection-statements corresponding to suitable subsets of phase space of a mechanical system, together with a time evolution of these projection-statements corresponding with usual classical mechanics. 

Given a realistic quantum model of Nature, it would also be very interesting to demonstrate a rigorous and specific version of Heuristic \ref{classical cutoff}, analogous to the von Neumann entropy bound occurring in several realistic models; \cite{Bekenstein,Bousso_2000}.

\subsubsection{Is winning the lottery dangerous?}

Our theory of natural probability is based on the assertion that our perceptions are described by classical projection-statements. In this view, we are largely classical creatures, described by classical projection-statements. The justification for this assertion is that information we perceive appears to be independently accessible in many places at once. It does, however, beg the question: \emph{Why do we only perceive classical projection statements?} 

Heuristic \ref{classical cutoff}  asserts that, for each observer there is some cutoff probability $p_O$ so that, for $O$ to be a classical projection-statement describing the observer, $\prob(O)>p_O$. In effect, Heuristic \ref{classical cutoff} implies that if $\prob(O)<p_O$, the classical predictions from the statement $O$ would not be reliable, because they would be swamped by quantum noise. Such a mechanism for explaining observed frequencies has been proposed by Hanson, \cite{Hanson_2003}, who used the language of large worlds `mangling' small worlds.  As animals, we are continually relying on classical predictions --- not only to find food and move safely through our surroundings, but also so that our body functions as usual. This raises an apparently absurd question: is it \emph{dangerous} to be described by a low probability projection statement $O$? In classical probability, it may be foolish to act on low probability information, because you are most likely wrong about that information, but within natural probability, there is the additional possibility that this may also make independent predictions less reliable, including, in extreme cases, the prediction that your body will continue to function as expected.

\bibliographystyle{plain}
\bibliography{qref.bib}

\end{document}